\documentclass[reqno,a4paper]{amsart}

\usepackage{amsmath,amsthm,amssymb,amsfonts,enumerate,graphicx,color,pstricks,marginnote}
\numberwithin{equation}{section} 
\theoremstyle{plain}
\newtheorem{theo+}           {Theorem}      [section]
\newtheorem{prop+}  [theo+]  {Proposition}
\newtheorem{coro+}  [theo+]  {Corollary}
\newtheorem{lemm+}  [theo+]  {Lemma}
\newtheorem{defi+}  [theo+]  {Definition}
\newtheorem{conj+}  [theo+]  {Conjecture}

\theoremstyle{definition}
\newtheorem{rema+}  [theo+]  {Remark}
\newtheorem{prob+}  [theo+]  {Problem}
\newtheorem{exam+}  [theo+]  {Example}
\newtheorem{ass+} [theo+] {Assumption}

\newenvironment{theorem}{\begin{theo+}}{\end{theo+}}
\newenvironment{proposition}{\begin{prop+}}{\end{prop+}}
\newenvironment{corollary}{\begin{coro+}}{\end{coro+}}
\newenvironment{lemma}{\begin{lemm+}}{\end{lemm+}}

\newenvironment{assumption}  {\begin{ass+}}{\end{ass+}}

\newcommand{\ti}{\mathrm i}
\newcommand{\Id}{\operatorname{Id}}
\newcommand{\Ker}{\operatorname{Ker}}

\begin{document}

\baselineskip 18pt
\larger[2]
\title
[Nearest-neighbour correlation functions and Painlev\'e VI] 
{Nearest-neighbour correlation functions\\ for the supersymmetric XYZ spin chain\\ and Painlev\'e VI}
\author{Christian Hagendorf}
\thanks{The first author acknowledges support from the Fonds de la Recherche Scientifique (F.R.S.-FNRS) and the Wetenschappeleijk Onderzoek-Vlanderen (FWO) through the Belgian Excellence of Science (EOS) project no. 30889451 ``PRIMA -- Partners in Research on Integrable Models and Applications''.}
\address
{Institut de Recherche en Math\'ematique et Physique \\ Universit\'e catholique de Louvain\\
Chemin du Cyclotron 2\\
B-1348 Louvain-la-Neuve\\ Belgium}
\email{christian.hagendorf@uclouvain.be}
\author{Hjalmar Rosengren}
\thanks{The second author is supported by the Swedish Research Council, project no.\ 2020-04221.}
\address
{Department of Mathematical Sciences
\\ Chalmers University of Technology and University of Gothenburg\\SE-412~96 G\"oteborg, Sweden}
\email{hjalmar@chalmers.se}
\urladdr{http://www.math.chalmers.se/{\textasciitilde}hjalmar}

\begin{abstract}
We study nearest-neighbour correlation functions for the ground state of the supersymmetric XYZ spin chain with odd length and periodic boundary conditions. Under a technical assumption related to the $Q$-operator of the corresponding eight-vertex model, we show that they can be expressed exactly in terms of the Painlev\'e VI tau functions $s_n$ and $\bar s_n$ introduced by Bazhanov and Mangazeev. Furthermore, we give an interpretation of the correlation functions in terms of the Painlev\'e VI Hamiltonian.
\end{abstract}

\maketitle

\section{Introduction}

In the theory of lattice models, solvability usually refers to situations when physically relevant quantities can be computed exactly in the infinite-lattice limit. A famous example is Baxter's computation of the ground-state energy per site for the XYZ spin chain \cite{b2}. By contrast, it is much more unusual to encounter exact results for finite-size systems.  In the context of the XYZ chain, one instance of this phenomenon was discovered by Stroganov \cite{st}. Baxter had already observed that, if the model's parameters satisfy
\begin{subequations}
\label{jss}
\begin{equation}J_xJ_y+J_xJ_z+J_yJ_z=0
\end{equation}
and
\begin{equation}
  J_x+J_y+J_z>0,
\end{equation}
\end{subequations}%
then the ground-state energy per site of the infinite chain takes the remarkably simple form
\begin{equation}\label{gse}-\frac{J_x+J_y+J_z}2.\end{equation}
Stroganov found empirically that \eqref{gse} seems to hold exactly also for finite chains of odd length with periodic boundary conditions.
 
One conceptual explanation of why the condition \eqref{jss} is  special was given by Fendley and the first author \cite{fh}. They proved that in this case the XYZ chain is supersymmetric  in the sense that, roughly speaking, the Hamiltonian can be expressed as an anticommutator of nilpotent operators. This property was used in \cite{hl} to give a rigorous proof of Stroganov's observation. 

When $J_x=J_y$, the XYZ chain reduces to the XXZ chain. It is customary to take $J_x=J_y=1$ and write $J_z=\Delta$. The condition \eqref{jss} is then $\Delta=-1/2$. There is a large literature on the XXZ chain with $\Delta=-1/2$ and its relations to combinatorial objects such as loop configurations, plane partitions and alternating sign matrices. We mention here only the Razumov--Stroganov conjecture \cite{rs1}, which was eventually proved by Cantini and Sportiello \cite{cs}.  Generalizing this work to the XYZ chain is difficult and  relations to combinatorics are still not well understood. However, it seems that three-colourings should play a role \cite{h1,h2,r3}.

An  intriguing aspect of the supersymmetric XYZ chain is its relations to the Painlev\'e VI equation, which are  present even for finite chains. Bazhanov and Mangazeev studied the ground-state eigenvalue of the $Q$-operator, in Stroganov's setting of odd length $L=2n+1$ and periodic boundary conditions \cite{bm,bm1,bm4}. They found that, at special values of the spectral parameter, this eigenvalue can be expressed in terms of certain polynomials denoted $s_n$ and $\bar s_n$. It was observed that these polynomials satisfy recursions that can be used to identify them with tau functions of Painlev\'e VI. These recursions were proved in \cite{r4}.  It also appears that the same polynomials are directly related to the eigenvector. In an appropriate normalization and ignoring elementary factors, its square norm is conjecturally given by $s_ns_{-n-1}$, its component with all spins equal given by  $s_n$ and its component with all but one spin equal given by  $\bar  s_n$. Partial proofs of these results have been given in \cite{bh,zj}. 
 
In the present work, we study nearest-neighbour correlation functions for the ground state of the supersymmetric XYZ chain, still in Stroganov's setting. We present an overview of our results on these correlation functions in \S 2. We find that they can be expressed in terms of a single quantity that we denote $f_n$, see Proposition \ref{cfp}. In our main result, Theorem \ref{mt}, we express $f_n$ in terms of the combination $\bar s_n\bar s_{-n-1}/s_ns_{-n-1}$ of Painlev\'e tau functions. Our proof is incomplete, as it is based on a technical assumption related to the eigenvalue of the $Q$-operator. In Theorem \ref{fpqp}, we give  a more direct interpretation of $f_n$ as a Painlev\'e VI Hamiltonian, evaluated at a solution of another instance of Painlev\'e VI, with shifted parameters. 

Our methods are similar to those of Stroganov \cite{st,st2}, who computed the correlation functions in the XXZ limit case, but the details are more involved. We present the details of the computation in \S 3. The XYZ chain can be parametrized by the crossing parameter $\eta$ and the elliptic nome $\tau$, where the supersymmetric case is $\eta=\pi/3$. Using the  Hellmann--Feynman theorem, we can express the quantity $f_n$ in terms of the ground-state energy per site $\varepsilon$, see \eqref{feg}. We then exploit the relation between the XYZ spin chain and the eight-vertex model to deduce an expression for $f_n$ in terms of the  transfer-matrix eigenvalue $\lambda(u)$, see Proposition \ref{fnp}. Next, we apply Baxter's $TQ$-relation, which relates the transfer matrix and the $Q$-operator. It is at this point that we need to make an assumption, Assumption~\ref{ass}. Although $Q$-operators are known both for generic $\eta\neq\pi/3$ and for $\eta=\pi/3$, we do not know any $Q$-operator that is analytic in a neighbourhood of $\eta=\pi/3$. Thus, it is not clear that we can differentiate the $TQ$-relation with respect to $\eta$. However, assuming that this is allowed leads to an expression for $f_n$ in terms of the $Q$-operator eigenvalue $q(u)$, see Corollary \ref{lxc}. In the case of the XXZ chain, one can write down $q(u)$ explicitly and  easily finish the computation \cite{st,st2}. For the XYZ chain, most of the work remains. The key fact is a new differential-difference equation relating $q(u)$ and $q(u+\pi)$, see Theorem \ref{ddt}. We also need some properties of the polynomials $s_n$ and $\bar s_n$ proved in \cite{r4}.

After having proved Theorem \ref{mt}, we turn to Theorem \ref{fpqp} in \S 4. Here, the key fact is the factorization \eqref{hpf} of the Painlev\'e VI Hamiltonian for special parameter values, which is also related to classical solutions of Painlev\'e VI.

We relegate several technical points to the appendices. In Appendix \ref{tfs}, we collect a number of theta function identities that we use throughout the main text. In Appendix \ref{sp}, we present several properties of the polynomials $s_n,\bar s_n$. Finally, we discuss a derivation of the nearest-neighbour correlation functions in the infinite-lattice limit, starting from Baxter's formula for the ground-state energy per site, in Appendix \ref{ills}. It is interesting to note that our expression for $f_n$, given in Theorem \ref{mt}, splits naturally as an elementary term corresponding to the infinite-lattice limit and a term involving Painlev\'e tau functions that gives the finite-length correction.

We conclude this introduction by mentioning several open problems. Clearly, it would be interesting to find explicit expressions for other correlation functions. For instance, the one-point correlation functions are known in the infinite-lattice limit \cite{bk,jmn}, but an exact finite-size result for the supersymmetric case remains to be found. Another natural candidate to investigate is the finite-size emptiness formation probability, whose expression is explicitly known for the limit of the XXZ chain at $\Delta=-1/2$ \cite{ca}.

One major obstruction to studying the supersymmetric XYZ chain seems to be that its ground-state eigenvectors are not known explicitly. By contrast, in the  XXZ case one can write down integral formulas for all their components  \cite{rsz}. Importantly for many applications, these formulas extend to the inhomogeneous six-vertex model.
 If one had similar expressions for the eight-vertex model, one could presumably complete the proof of many properties of the ground state described in \cite{bm4,bh,rs,zj}. They might also be useful for proving our Assumption \ref{ass} and for computing more general correlation functions. We mention in this context the paper \cite{fwz}, which contains integral formulas of the desired type for the eight-vertex-solid-on-solid model.

One important limit of the supersymmetric XYZ chain is obtained by letting $n$ and $\tau/\ti$ tend to infinity with $e^{\pi\ti\tau}\sim n^{-2/3}$  \cite{bm}. This limit is believed to be related to the sine-Gordon model at its supersymmetric point and to polymers on a cylinder \cite{l,fs}. In this limit, Painlev\'e VI should degenerate to Painlev\'e III. It would be interesting to compute the corresponding limit of the  correlation functions.

The present work provides further examples of the relation between the supersymmetric XYZ spin chain and the Painlev\'e VI equation discovered in \cite{bm1}. We stress that we have no conceptual explanation for this relation. Both our incomplete proof that the polynomials $s_n$ and $\bar s_n$ appear in the correlation functions and the proof that they are Painlev\'e tau functions \cite{r4} emerge like miracles at the end of long computations.
   
Another phenomenon that calls for further explanation is the relation between the classical and quantum versions of Painlev\'e VI. To discuss this, let us first recall the elliptic form of Painlev\'e VI, given in \cite{ma} as
\begin{equation}
   \label{cpvi}\frac{d^2q}{d\tau^2}=-\frac 1{4\pi^2}\sum_{j=0}^3\beta_j\wp'(q+\omega_j|1,\tau),
\end{equation}
where $\wp$ is the Weierstrass elliptic function, $\beta_j$ are parameters and
$$(\omega_0,\omega_1,\omega_2,\omega_3)=\left(0,\frac 12,\frac\tau2,\frac{\tau+1}2\right). $$
Writing $p=dq/d\tau$, it is equivalent to the Hamiltonian system
$$\frac{d q}{d \tau}=\frac{\partial H}{\partial p},\qquad \frac{dp}{d\tau}=-\frac{\partial H}{\partial q}, $$
   where
$$H=\frac{p^2}{2}+V(q,\tau)$$
and 
\begin{equation}\label{dtv}
V(q,\tau)=\frac 1{4\pi^2}\sum_{j=0}^3\beta_j\wp(q+\omega_j|1,\tau)
\end{equation}
is the Darboux--Treibich--Verdier potential \cite{d,tv,v}.
 
The  corresponding Schr\"odinger equation
  \begin{equation}
  \label{qpvi}
  2\pi\ti \psi_\tau=\frac 12\psi_{xx}-V(x,\tau)\psi
\end{equation}
has been called the quantum Painlev\'e VI equation (if $\tau\in\ti\mathbb R$ it is a Schr\"odinger equation in  imaginary time, that is, a heat equation). There is a direct link between classical and quantum Painlev\'e VI based on Lax pairs, the so called quantum Painlev\'e--Calogero correspondence \cite{zz}, see also \cite{cd,n,su1,su2}.
 
In \cite{bm}, it was found that the $Q$-operator eigenvalue satisfies the non-station\-ary Lam\'e equation, which is the special case of \eqref{qpvi} where $\beta_0=\beta_1=n(n+1)/2$ and $\beta_2=\beta_3=0$. More generally, in \cite{r4} the second author constructed solutions to \eqref{qpvi} with any values $\beta_j=k_j(k_j+1)/2$, $k_j\in\mathbb Z$. For such parameters, the potential \eqref{dtv} has the so-called finite-gap property \cite{tv,v}. Specializing the variable $x$ to a half-period, the resulting functions of $\tau$ are tau functions of classical Painlev\'e VI in the Picard class, that is, $\beta_j=l_j^2/2$, $l_j\in\mathbb Z$. This link between quantum and classical Painlev\'e VI seems different from the quantum Painlev\'e--Calogero correspondence of \cite{zz}. It would be interesting to know if there is a hidden connection.

\section{Statement of results}
\label{srs}

\subsection{Nearest-neighbour correlation functions} We write $V=\mathbb C|{\uparrow}\rangle \oplus \mathbb C|{\downarrow}\rangle$.
The Pauli matrices acting on $V$ are
$$\sigma^x=\left[\begin{matrix}0&1\\1&0\end{matrix}\right],\qquad  \sigma^y=\left[\begin{matrix}0&-\ti \\ \ti &0\end{matrix}\right],\qquad \sigma^z=\left[\begin{matrix}1&0\\0&-1\end{matrix}\right]. $$
The XYZ spin chain is defined by the Hamiltonian
\begin{equation}\label{ham}\mathbf H=-\frac 12\sum_{j=1}^L\left(J_x\,\sigma_j^x\sigma_{j+1}^x+J_y\,\sigma_j^y\sigma_{j+1}^y+J_z\,\sigma_j^z\sigma_{j+1}^z\right) \end{equation}
acting on $V^{\otimes L}$.
Here, $J_x$, $J_y$ and $J_z$ are real anisotropy parameters and the lower indices of the Pauli matrices indicate on which tensor factor they act. We will only consider periodic boundary conditions, that is, the index $L+1$ should be understood as $1$. Moreover, we will always assume that $L$ is odd and write $L=2n+1$.

We are interested in the supersymmetric case \eqref{jss}. We write the lowest eigenvalue of $\mathbf H$ as $L\varepsilon$, where $\varepsilon$ is the ground-state energy per site. As was mentioned in the introduction, it is known that
\begin{equation}\label{gsf}
  \varepsilon=-\frac{J_x+J_y+J_z}2.
\end{equation}

Let $\Psi$ be an eigenvector of $\mathbf H$ with eigenvalue $L\varepsilon$.
We write
$$\langle \mathbf A\rangle=\frac{\langle\Psi|\mathbf A|\Psi\rangle}{\langle\Psi|\Psi\rangle} $$
for the expectation value of an observable $\mathbf A$ with respect to the corresponding ground-state. Here, $\langle \Phi|\Phi'\rangle$ denotes the standard Hermitian scalar product of two vectors $\Phi,\Phi'\in V^{\otimes L}$. We are interested in the nearest-neighbour correlation functions
\begin{equation}\label{cfd}C^x=\langle \sigma_j^x\sigma_{j+1}^x\rangle,\qquad  C^y=\langle \sigma_j^y\sigma_{j+1}^y\rangle,\qquad C^z=\langle \sigma_j^z\sigma_{j+1}^z\rangle.\end{equation}

\begin{lemma}\label{igsl}
The correlation functions \eqref{cfd} are independent of the choice of ground-state vector $\Psi$, as well as of the index $j$.
\end{lemma}

\begin{proof}
Write $V^{\otimes L}=W^+\oplus W^-$, where $W^+$ and $W^-$ are spanned by states with, respectively, an even and an odd number of down spins.
This decomposition is preserved by $\mathbf H$. It is proved in \cite{hl} that the ground-state eigenspace of $\mathbf H\big|_{W^+}$ is one-dimensional. Let $\Psi^+$ be a vector in this space, normalized so that $\langle\Psi^+|\Psi^+\rangle=1$. Then, the whole eigenspace is spanned by $\Psi^+$ and
$\Psi^-=\mathbf F\Psi^+\in W^-$, where $\mathbf F=\prod_{j=1}^L\sigma^x_j$ is the spin-flip operator.
Let $\Psi=\alpha\Psi^++\beta\Psi^-$ be an arbitrary ground-state vector. By orthogonality,
$$\langle \Psi|\Psi\rangle=|\alpha|^2 \langle\Psi^+|\Psi^+\rangle+|\beta|^2 \langle\Psi^-|\Psi^-\rangle=|\alpha|^2+|\beta|^2.$$
Moreover, for $a\in\{x,y,z\}$,
\begin{align*}\langle\Psi|\sigma_j^a\sigma_{j+1}^a|\Psi\rangle&=|\alpha|^2\langle\Psi^+|\sigma_j^a\sigma_{j+1}^a|\Psi^+\rangle
+|\beta|^2\langle\Psi^-|\sigma_j^a\sigma_{j+1}^a|\Psi^-\rangle\\
&=(|\alpha|^2+|\beta|^2)\langle\Psi^+|\sigma_j^a\sigma_{j+1}^a|\Psi^+\rangle,
 \end{align*}
where we used in the first step that the operator $\sigma_j^a\sigma_{j+1}^a$ preserves the spaces $W^\pm$ and in the second step that it commutes with $\mathbf F$. This shows that the correlation functions \eqref{cfd} do not change if we replace $\Psi$ by $\Psi^+$. The final statement follows since $\Psi^+$ is translation invariant \cite{hl} (or, alternatively, since by the first statement we could replace $\Psi^+$ by any one of its translates).  
\end{proof}

Our goal is to compute the correlation functions \eqref{cfd} exactly.
As a first step, we express them in terms of a single quantity that we denote $f_n$.
It will be convenient to introduce the normalized discriminant
 \begin{equation}\label{capz}Z=\frac{(J_x-J_y)^2(J_y-J_z)^2(J_x-J_z)^2}{J_x^2J_y^2J_z^2}=-4\frac{(J_x+J_y+J_z)^3}{J_xJ_yJ_z}-27, \end{equation}
 where the second equality depends on \eqref{jss}. 

\begin{proposition}\label{cfp}
For the supersymmetric periodic XYZ spin chain of odd length $L=2n+1$, 
we can write
\begin{equation}\label{caj}C^{a}=1+\frac{J_xJ_yJ_z}{J_a^2(J_x+J_y+J_z)}\,f_n,\qquad a\in\{x,y,z\},\end{equation}
where
 $f_n$ is a rational function of $Z$. 
\end{proposition}

For small $n$, it is straightforward to compute the ground-state eigenvectors of the supersymmetric XYZ Hamiltonian and infer $f_n$. To give some examples, we find for $n=0,\dots,5$ the expressions
\begin{align*}
f_0&=0,\\
f_1&=1,\\
f_2&=\frac{Z+27}{Z+25},\\
f_3&=\frac{(Z+24)(Z+27)}{(Z+21)(Z+28)},\\
f_4&=\frac{Z^3+74Z^2+1807Z+14520}{Z^3+72Z^2+1701Z+13068},\\
f_5&=\frac{(Z+27)(Z^4+96Z^3+3420Z^2+53404Z+306735)}{(Z^2+44Z+429)(Z^3+77Z^2+1991Z+17303)}.
\end{align*}

Since the correlation functions only depend on the Hamiltonian up to normalization, 
we may use the parametrization
\begin{equation}\label{jsp}J_x=1+\zeta,\qquad J_y=1-\zeta,\qquad J_z=\frac{\zeta^2-1}{2}.
\end{equation}
Then, \eqref{caj} takes the form
\begin{equation}\label{cf}C^x=1-\frac{(1-\zeta)^2f_n}{\zeta^2+3},\quad C^y=1-\frac{(1+\zeta)^2f_n}{\zeta^2+3},\quad  C^z=1-\frac{4f_n}{\zeta^2+3},\end{equation}
where $f_n$ is a function of 
\begin{equation}\label{zz}Z=\frac{\zeta^2(\zeta^2-9)^2}{(\zeta^2-1)^2}.\end{equation}
The symmetry of $Z$ under permutations of the anisotropy parameters corresponds to
the fact that $Z(\zeta)=Z(\zeta')$ if and only if
$\zeta'\in\{\pm\zeta,\pm\gamma,\pm \delta\}$, where
$$\gamma=\frac{\zeta+3}{\zeta-1},\qquad \delta=\frac{\zeta-3}{\zeta+1}.$$
That is, 
\begin{equation}\label{fns}f_n(\pm\zeta)=f_n(\pm\gamma)=f_n(\pm\delta). \end{equation}
Indeed, it is easy to check that the transformations $\zeta\mapsto\zeta'$ permute
the parameters \eqref{jsp} up to normalization.

\subsection{Correlation functions and tau functions}
Our main result gives an explicit formula for  $f_n$ in terms of the polynomials $s_n$ and $\bar s_n$ introduced by Bazhanov and Mangazeev \cite{bm}. As was conjectured in \cite{bm1} and proved in \cite{r4}, they can be identified with tau functions of Painlev\'e VI and, hence, satisfy a Toda-type recursion. For $s_n=s_n(z)$, $n\in\mathbb Z$,  this recursion takes the form
\begin{multline}\label{srec}
  8(2n+1)^2s_{n+1}s_{n-1}+2z(z-1)(9z-1)^2(s_n''s_n-(s_n')^2)+2(3z-1)^2(9z-1)s_ns_n'\\-\big(4(3n+1)(3n+2)+n(5n+3)(9z-1)\big)s_n^2=0
\end{multline}
with starting values $s_0=s_1=1$. To obtain $\bar s_n$, one should replace $n(5n+3)$ by $(n-1)(5n+4)$ and the starting values by $\bar s_0=1$, $\bar s_1=3$. It was conjectured in \cite{mb} that, in a certain natural normalization of the eigenvector $\Psi$, its square norm is essentially given by $s_n(\zeta^{-2})s_{-n-1}(\zeta^{-2})$. A partial proof of this conjecture is given in \cite{zj}. Thus, it is natural to expect this product in the denominator of $f_n$. More surprisingly, it appears that the numerator can be expressed in terms of $\bar s_n\bar s_{-n-1}$. Although we have checked the following result for chains up to length $L=11$, we have only proved it under a technical assumption that will be explained below.

\begin{theorem}\label{mt}
If \emph{Assumption \ref{ass}} holds, then
  \begin{equation}\label{fd}f_n=\frac{(\zeta^2+3)(\zeta^2-3)}{(\zeta^2-1)^2}-\frac{2\zeta^2(\zeta^2+3)}{(2n+1)^2(\zeta^2-1)^2}\frac{\bar s_n(\zeta^{-2})\bar s_{-n-1}(\zeta^{-2})}{s_n(\zeta^{-2})s_{-n-1}(\zeta^{-2})}. \end{equation}
\end{theorem}

By \eqref{fns}, one can obtain seemingly different expressions for $f_n$ by replacing $\zeta$ in \eqref{fd} with $\gamma$ or $\delta$. 

The decomposition of $f_n$ into two terms has a  natural interpretation. Namely,  in Appendix \ref{ills} we 
argue that $f_\infty=\lim_{n\rightarrow\infty}f_n$ is given by
\begin{equation}\label{finfe}f_\infty=\begin{cases}\displaystyle\frac{(\zeta^2+3)(\zeta^2-3)}{(\zeta^2-1)^2}, & |\zeta|\geq 3,\\[3mm]
\displaystyle\frac{(\delta^2+3)(\delta^2-3)}{(\delta^2-1)^2}=-\frac{(\zeta^2+3)(\zeta^2+6\zeta-3)}{8(\zeta-1)^2}, & -3\leq \zeta\leq 0,\\[3mm]
\displaystyle\frac{(\gamma^2+3)(\gamma^2-3)}{(\gamma^2-1)^2}=-\frac{(\zeta^2+3)(\zeta^2-6\zeta-3)}{8(\zeta+1)^2},& 0\leq \zeta\leq 3.\\
\end{cases} \end{equation}
Hence, when $|\zeta|\geq 3$,
the first term in \eqref{fd} gives the infinite-lattice limit and the second term the  finite length correction. The variants of \eqref{fd} with $\zeta$ replaced by $\gamma$ and $\delta$ have a similar interpretation in the other  parameter regimes. 
In Appendix \ref{ills}, we deduce \eqref{finfe} from Baxter's formula for the ground-state energy per site in the infinite-lattice limit. We do not know how to prove \eqref{finfe} directly from \eqref{fd}. In Figure \ref{ffinf}, we illustrate the convergence of $f_n$ to $f_\infty $.
\begin{figure}[h]
  \centering
  \includegraphics[width=.65\textwidth]{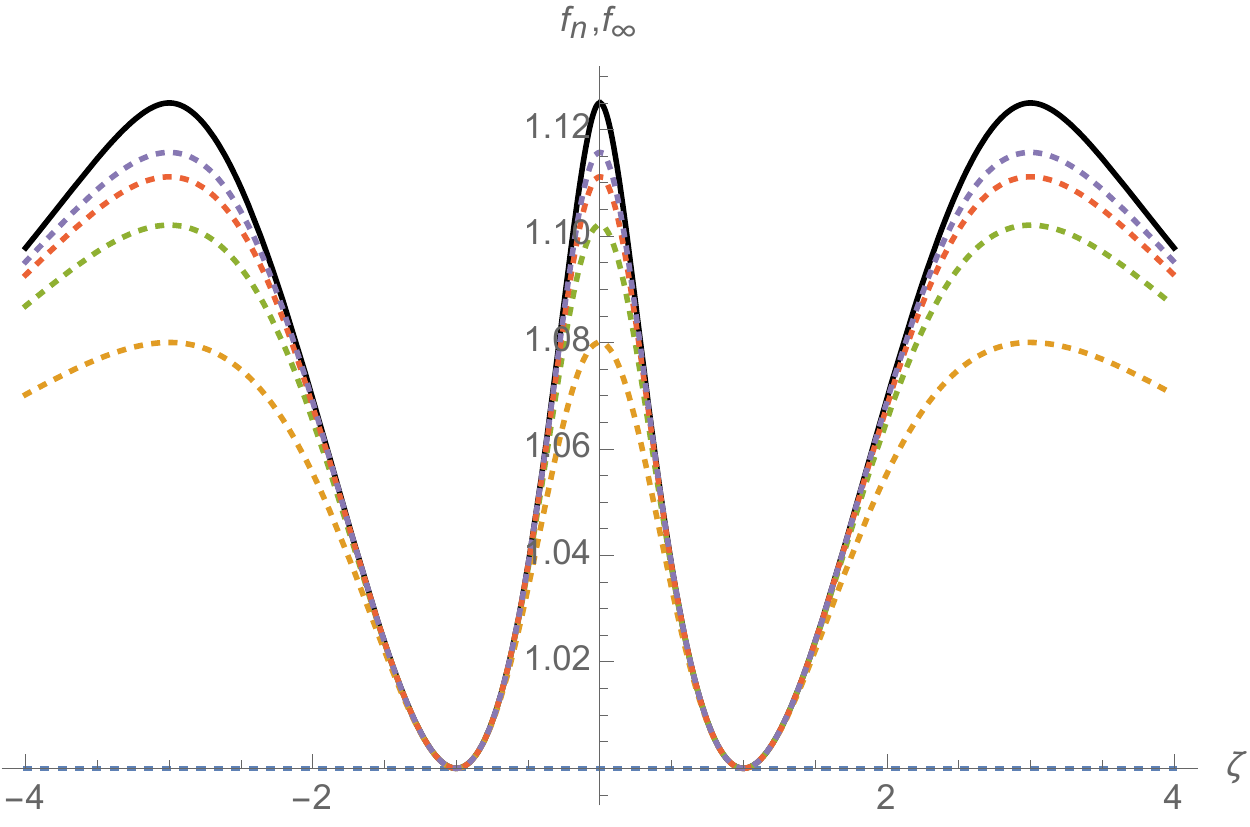}
  \caption{Plots of $f_n$ for $n=1,\dots,5$ (dotted lines) and $f_\infty$ (black solid line) as functions of $\zeta$.}
  \label{ffinf}
\end{figure}

To see that Theorem \ref{mt} gives the correct result for the XXZ chain, we replace $\zeta$ by $\gamma$ in \eqref{fd} and let $\zeta=0$, $\gamma=-3$. Then,  the parameters \eqref{jsp} reduce to $J_x=J_y=1$, $J_z=-1/2$. At the corresponding value $z=1/9$, the recursion \eqref{srec} for $s_n$ reduces to
$$2(2n+1)^2s_{n+1}s_{n-1}= (3n+1)(3n+2)s_n^2$$
and the recursion for $\bar s_n$ is identical. It easily follows that  $\bar s_n(1/9)=3^n s_n(1/9)$. This leads to $f_n=9/8-9/(8L^2)$ and we recover Stroganov's results \cite{st,st2}:
$$C^x=C^y=\frac{5}{8}+\frac 3{8L^2},\qquad C^z=-\frac 12+\frac{3}{2L^2}. $$ One can also verify Theorem \ref{mt} in the limit  $\zeta\rightarrow\infty$, which corresponds to $J_x=J_y=0$ and $J_z =1/2$ (after rescaling the Hamiltonian appropriately). Indeed, it follows from \cite[Thm.\ 4.1]{rsa}  that $s_n(0)\neq 0$ for all $n$, $\bar s_n(0)\neq 0$ for $n\geq 0$ and that $\bar s_{-n-1}(z)$ is divisible by $z^n$ for $n>0$. Hence, the second term in \eqref{fd} tends to zero (except when $n=0$) and we obtain $C^x=C^y=0$, $C^z=1$. This is trivial to check directly since the ground-state vector can be chosen as  $\Psi=\left|\uparrow\right\rangle^{\otimes L}$.

\subsection{The Painlev\'e Hamiltonian} 

The expression \eqref{fd}  can be interpreted as the Painlev\'e VI Hamiltonian,  evaluated at a solution to the same equation with different parameters. To explain this, we recall the algebraic form of Painlev\'e VI,
\begin{align}\notag
\frac{d^2q}{dt^2}&=\frac 12\left(\frac 1q+\frac 1{q-1}+\frac 1{q-t}\right)
\left(\frac{dq}{dt}\right)^2-\left(\frac 1t+\frac 1{t-1}+\frac 1{q-t}\right)\frac{dq}{dt}\\ 
\label{py}&\quad +\frac{q(q-1)(q-t)}{t^2(t-1)^2}\left(\alpha+\beta\frac t{q^2}+\gamma\frac{t-1}{(q-1)^2}+\delta\frac{t(t-1)}{(q-t)^2}\right).\end{align}
As in \cite{ny}, we will write
$$\alpha=\frac{\alpha_1^2}{2},\qquad \beta=-\frac{\alpha_4^2}{2},\qquad
\gamma=\frac{\alpha_3^2}{2},\qquad \delta=\frac{1-\alpha_0^2}{2}$$
and introduce $\alpha_2$ so that
\begin{equation}\label{ac}
 \alpha_0+\alpha_1+2\alpha_2+\alpha_3+\alpha_4=1.
\end{equation}
There is a change of variables that takes \eqref{py} to the elliptic form \eqref{cpvi}, where the correspondence of parameters is
   $$(\beta_0,\beta_1,\beta_2,\beta_3)=
   \left(\frac{\alpha_1^2}{2},\frac{\alpha_4^2}2,\frac{\alpha_3^2}2,\frac{\alpha_0^2}2\right). $$
 
With
$$p=\frac 12\left(\frac{\alpha_4}q+\frac{\alpha_3}{q-1}+\frac{\alpha_0-1}{q-t}
+\frac{t(t-1)}{q(q-1)(q-t)}\frac{dq}{dt} \right),$$
the Painlev\'e equation \eqref{py} is equivalent to the non-stationary Hamiltonian system \cite{m}
\begin{equation}\label{ph}
t(t-1)\frac{dq}{dt}=\frac{\partial H}{\partial p},\qquad t(t-1)\frac{dp}{dt}=-\frac{\partial H}{\partial q},
\end{equation}
where
\begin{multline*}
H=q(q-1)(q-t)p^2-\big\{(\alpha_0-1)q(q-1)+\alpha_3q(q-t)+\alpha_4(q-1)(q-t)\big\}p\\
+\alpha_2(\alpha_1+\alpha_2)(q-t).
\end{multline*}

Starting from one solution of Painlev\'e VI, one can obtain further solutions by applying B\"acklund transformations \cite{o}. For instance, if $(q,p)$ solves \eqref{ph}, then 
\begin{equation}\label{tb}
T(q)=\frac{tp(pq-\alpha_4)}{(pq+\alpha_2)(pq+\alpha_1+\alpha_2)}
\end{equation}
defines another solution with shifted parameters
$$ T(\alpha_0,\alpha_1,\alpha_2,\alpha_3,\alpha_4)
=(\alpha_0-1,\alpha_1,\alpha_2+1,\alpha_3-1,\alpha_4).$$ 
We will apply this transformation to the algebraic solution $(q_0,p_0)$ with $$\alpha_0=\alpha_1=\alpha_3=\alpha_4=0,\qquad \alpha_2=\frac 12,$$ defined by the parametrization
\begin{equation}\label{sps}q_0=\frac{s(s+2)}{2s+1},\qquad 
p_0=-\frac{2s+1}{2(s-1)(s+2)},\qquad
t= \frac{s(s+2)^3}{(2s+1)^3}.\end{equation}
We can then give the following curious interpretation of the function $f_n$. 

\begin{theorem}\label{fpqp}
The expression \eqref{fd} can be written as
\begin{align*}
 f_n&=\frac{(s^2+s+1)(s^2+4s+1)}{(s+1)^4}-\frac{4(2s+1)^3(s^2+4s+1)}{(2n+1)^2s(s+1)^4}\left(H'+\frac{(2n+1)^2}{4}\,t\right),
\end{align*} 
where $t$ is as in \eqref{sps},
\begin{equation}\label{zx}
  Z=\frac{(s-1)^4(s+2)(2s+1)}{s(s+1)^4},
\end{equation}
and where $H'=H'(q_n,p_n,t)$ is the Painlev\'e VI Hamiltonian with  parameters
\begin{equation}\label{shpp}
(\alpha_0,\alpha_1,\alpha_2,\alpha_3,\alpha_4)
=\left(\frac 12-n,0,\frac 12+n,-\frac 12-n,0\right),
\end{equation}
evaluated at  $(q_n,p_n)=T^n(q_0,p_0)$, which solves Painlev\'e VI with parameters $$ (\alpha_0,\alpha_1,\alpha_2,\alpha_3,\alpha_4)
=\left(-n,0,\frac 12+n,-n,0\right).$$
   \end{theorem}
   
    To write Theorem \ref{fpqp} in terms of the parametrization \eqref{jsp} one should equate \eqref{zz} and \eqref{zx}. This is a sextic equation, whose solutions can be written
    \begin{equation}\label{zgds}\frac{(s+2)(2s+1)}{s}\in\{\zeta^2,\gamma^2,\delta^2\}.
    \end{equation}
    Theorem \ref{fpqp} holds with each of the corresponding six choices for $s$.


\section{Computation of the correlation functions}

\subsection{The eight-vertex model}\label{evmss}

To compute $f_n$, we will exploit the relation between the XYZ spin chain and the eight-vertex model \cite{b2}. With respect to the basis 
 $|{\uparrow\uparrow}\rangle$, $|{\uparrow\downarrow}\rangle$, $|{\downarrow\uparrow}\rangle$, $|{\downarrow\downarrow}\rangle$  of $V^{\otimes 2}$, the $R$-matrix of
 the eight-vertex model is given by
$$R=\left[\begin{matrix}a&0&0&d\\ 0&b&c&0\\ 0&c&b&0\\ d&0&0&a\end{matrix}
\right], $$
where $a$, $b$, $c$, $d$ are  Boltzmann weights. The transfer matrix is the operator on $V^{\otimes L}$ defined by $$\mathbf T=\operatorname{Tr}_{0}(R_{01}R_{02}\dotsm R_{0L}), $$ where the index $0$ refers to an additional auxiliary copy of $V$. The parameter combinations  
\begin{equation}\label{zg}
  \zeta=\frac{cd}{ab},\qquad \Gamma=\frac{a^2+b^2-c^2-d^2}{2ab}
\end{equation}
play a special role. Indeed, if $(a',b',c',d')$ is another set of parameters such that $\zeta'=\zeta$ and $\Gamma'=\Gamma$, then the corresponding transfer matrices $\mathbf T$ and $\mathbf T'$ commute. 

An important step in Baxter's solution of the eight-vertex model is to parametrize the Boltzmann weights by theta functions. In our conventions,
\begin{subequations}\label{bw}
\begin{align}
a&=\rho\,\theta_4(2\eta|2\tau)\theta_4(u-\eta|2\tau)\theta_1(u+\eta|2\tau),\\
b&=\rho\,\theta_4(2\eta|2\tau)\theta_1(u-\eta|2\tau)\theta_4(u+\eta|2\tau),\\
c&=\rho\,\theta_1(2\eta|2\tau)\theta_4(u-\eta|2\tau)\theta_4(u+\eta|2\tau),\\
d&=\rho\,\theta_1(2\eta|2\tau)\theta_1(u-\eta|2\tau)\theta_1(u+\eta|2\tau);
\end{align}
\end{subequations}
see \eqref{jt} for the  definition of Jacobi theta functions. The spectral parameter $u$  and the crossing parameter $\eta$ can be taken as complex, and the modular parameter $\tau$ as complex in the upper half-plane. Without loss of generality, we fix the normalization $\rho$ as $$\rho=\frac{2}{\theta_2(0|\tau)\theta_4(0|2\tau)},$$ which implies
\begin{equation}\label{phf}
  a+b=\frac{\theta_1(2\eta|\tau)}{\theta_1(\eta|\tau)}\,\theta_1(u|\tau). \end{equation}   
As a function of $u$, the transfer matrix satisfies the quasi-periodicity and reflection symmetry
\begin{equation}\label{tqp}
  \mathbf T(u+\pi)=\mathbf T(-u)=(-1)^L\mathbf T(u),\qquad \mathbf T(u+\pi\tau)=(-1)^L e^{-\ti L(2u+\pi\tau)}\mathbf T(u).
\end{equation}

In the parametrization \eqref{bw}, the quantities \eqref{zg} are independent of $u$. More precisely \cite[Eqs.\ (2.6)--(2.7)]{rs}, 
\begin{equation}\label{zge}\zeta=\frac{\theta_1(2\eta|2\tau)^2}{\theta_4(2\eta|2\tau)^2},
\qquad \Gamma=\frac{\theta_2(2\eta|2\tau)\theta_3(2\eta|2\tau)\theta_4(0|2\tau)^2}{\theta_2(0|2\tau)\theta_3(0|2\tau)\theta_4(2\eta|2\tau)^2}. \end{equation}
Hence, the transfer matrices $\mathbf T(u)$ form a commuting family for fixed $\eta$ and $\tau$. Moreover, the XYZ Hamiltonian is contained within this  family in the sense that
\begin{equation}\label{thr}
  \mathbf H=\frac{L(a_u+c_u)}{2b_u}\,\Id-\frac{a}{b_u} \,\mathbf T^{-1}\mathbf T_u\bigg|_{u=\eta},
\end{equation}
where the subscripts denote $u$-derivatives and the parameters are related through
\begin{equation}\label{aip}
  J_x=1+\zeta,\qquad J_y=1-\zeta,\qquad J_z=\Gamma.
\end{equation}
As was noted by Baxter \cite{b2}, when $\eta=\pi/3$ the supersymmetry condition \eqref{jss} holds. That is, $\Gamma=(\zeta^2-1)/2$ and we recover \eqref{jsp}.

Let us choose $\Psi=\Psi^+$ as in the proof of Lemma \ref{igsl}. It belongs to the sector $W^+$ spanned by states with an even number of down spins. Since $W^+$ is preserved by $\mathbf H$, the ground-state eigenvalue of $\mathbf H\big|_{W^+}$ remains single on $W^+$ for $\eta$ near $\pi/3$. We will still denote this eigenvalue by  $L\varepsilon$, where we know from \eqref{gsf} that
\begin{equation}\label{se}
  \varepsilon\Big|_{\eta=\pi/3}=-\frac{(\zeta^2+3)}{4}.
\end{equation}
We extend $\Psi$ to an eigenvector that depends analytically on $\eta$. It is also an eigenvector of the transfer matrices $\mathbf T(u)$; the eigenvalue will be denoted $\lambda$. It was conjectured in \cite{st} and proved in \cite{hl} that
\begin{equation}\label{ese}
  \lambda\Big|_{\eta=\pi/3}=\theta_1(u|\tau)^L.
\end{equation}
For generic $\eta$, there are no simple finite-size expressions for neither $\varepsilon$ nor $\lambda$, but \eqref{thr} gives the relation
 \begin{equation}\label{ler}\varepsilon=\frac{a_u+c_u}{2b_u}-\frac{a}{Lb_u}\frac{\lambda_u}{\lambda}\Bigg|_{u=\eta}.
 \end{equation} 

\subsection{Proof of  Proposition \ref{cfp}}
\label{cfpss}

We will now prove that Proposition \ref{cfp} holds, with 
\begin{equation}\label{feg}f_n=\frac{\varepsilon(2\varepsilon_\eta+\Gamma_\eta)}{\zeta\zeta_\eta-\Gamma_\eta}\Bigg|_{\eta=\pi/3}, \end{equation}
where the subscripts on the right-hand side denote $\eta$-derivatives.

We first note that \eqref{ham} implies
\begin{equation}\label{eh}L\varepsilon=\langle \mathbf H\rangle=-\frac L 2\left(J_x C^x+J_y C^y+J_zC^z\right).\end{equation}
The Hellmann--Feynman theorem asserts that, if $\mathbf H$ depends smoothly on a parameter, then
$$\langle \mathbf H\rangle'=\langle \mathbf H'\rangle, $$ where the prime denotes the derivative with respect to the parameter. Thus, we can write
$$\varepsilon'=-\frac 1 2\left(J_x' C^x+J_y' C^y+J_z'C^z\right). $$
The derivative may be taken with respect to either $\eta$ or $\tau$. Together with \eqref{eh}, we thus obtain a system of three linear equations for $C^x$, $C^y$ and $C^z$. Consider this system at the point $\eta=\pi/3$. Using \eqref{aip} and \eqref{se}, it can be written
\begin{equation}\label{cls}\left[\begin{matrix} 1+\zeta & 1-\zeta &\displaystyle \frac 12(\zeta^2-1)\\
\displaystyle\zeta_\eta&\displaystyle-\zeta_\eta&\Gamma_\eta\\
1 & -1 & \zeta\end{matrix}\right] \left[\begin{matrix}C^x\\C^y\\C^z\end{matrix}\right]= \left[\begin{matrix}\displaystyle\frac 12(\zeta^2+3)\\ \displaystyle-2\varepsilon_\eta\\ \zeta\end{matrix}\right].\end{equation}

The determinant of the coefficient matrix in \eqref{cls} is
$2(\Gamma_\eta-\zeta\zeta_\eta)$. By  Lemma~\ref{fdl}, 
\begin{equation}\label{zze}\zeta\zeta_\eta-\Gamma_\eta\Big|_{\eta=\pi/3}=6\chi\frac a{b_u}\Big|_{u=\eta=\pi/3}, 
\end{equation}
where
\begin{equation}\label{chi}\chi=\frac{\theta_1'(0|\tau)^2\theta_2(\pi/3|\tau)^2}{\theta_1(\pi/3|\tau)^2\theta_2(0|\tau)^2}.\end{equation}
This expression implies, in particular, that the determinant is non-zero for generic $\tau$. It is then easy to see that \eqref{cls} has the unique solution \eqref{cf}, where $f_n$ is given by \eqref{feg}. As we have noted, \eqref{cf} can equivalently be written as \eqref{caj}.

%
 
It remains to show that $f_n$ is a rational function of $Z$. Note that the eigenvalue equation for $\Psi$ is a linear system whose coefficients are linear in $(J_x,J_y,J_z)$. Hence, we may normalize $\Psi$ so that its coefficients are homogeneous polynomials in the anisotropy parameters. Then, \eqref{caj}
defines $f_n$ as a homogeneous rational function of $J_x$, $J_y$ and $J_z$. 
To complete the proof we need the following fact.

\begin{lemma}\label{cfsl}
Permuting the anisotropy parameters entails a corresponding permutation of the correlation functions. That is, if $(\tilde x,\tilde y,\tilde z)$ is any permutation of the symbols $(x,y,z)$ and we consider the correlations as functions of the anisotropy parameters, then
\begin{equation}
  \label{cfs}
  C^a(J_{\tilde x},J_{\tilde y},J_{\tilde z})=C^{\tilde a}(J_x,J_y,J_z),\qquad a\in\{x,y,z\}. 
\end{equation}
\end{lemma}

\begin{proof}
Let 
$\tilde{\mathbf H}$ be the operator obtained from the Hamiltonian $\mathbf H$ after replacing the parameters
$(J_x,J_y,J_z)$ with $(J_{\tilde x},J_{\tilde y},J_{\tilde z})$.
It is well-known that there exists $U\in\mathrm{SU}(2)$ such that $\sigma^{a}=\pm U\sigma^{\tilde a} U^{-1}$, $a\in\{x,y,z\}$. The signs involved in these formulas will be irrelevant.
Writing $\mathbf U=U^{\otimes L}$, we have
  $\tilde{\mathbf H}=\mathbf U\mathbf H \mathbf U^{-1}$. Given a ground-state eigenvector $\Psi$ of $\mathbf H$, we can obtain a ground-state eigenvector of $\tilde{\mathbf H}$ as  $\mathbf U\Psi$. By Lemma \ref{igsl}, we may use $\mathbf U\Psi$ to compute the left-hand side of \eqref{cfs}. That is, we are reduced to the identity
$$\frac{\langle\mathbf  U\Psi|\sigma_j^a\sigma_{j+1}^a|\mathbf U\Psi\rangle}{\langle \mathbf U\Psi|\mathbf U\Psi\rangle}
=\frac{\langle \Psi|\sigma_j^{\tilde a}\sigma_{j+1}^{\tilde a}|\Psi\rangle}{\langle \Psi|\Psi\rangle},
 $$
 which is clear  from the properties of $U$.
\end{proof}

It is easy to see that
the expressions \eqref{caj} are consistent with Lemma \ref{cfsl} only if $f_n$ is invariant under  permutation of the anisotropy parameters. It must then be
a rational expression in
$e_1=J_x+J_y+J_z$, $e_2=J_xJ_y+J_xJ_z+J_yJ_z$ and $e_3=J_xJ_yJ_z$.
Since $e_2=0$ and $f_n$ is homogeneous,
it is in fact  a function of $e_1^3/e_3$ or, equivalently, of $Z$.
This completes the proof of Proposition \ref{cfp}.

  \subsection{Expressions in terms of the transfer-matrix eigenvalue}

The next step is to express \eqref{feg} in terms of the transfer matrix eigenvalue $\lambda$.

\begin{lemma}\label{gel}
One has 
\begin{equation}\label{fni}\varepsilon_\eta+\frac{\Gamma_\eta} 2\Bigg|_{\eta=\pi/3}=-\frac{a}{Lb_u}\frac{\lambda_{u\eta}\lambda-\lambda_u\lambda_\eta}{\lambda^2}\Bigg|_{u=\eta=\pi/3}. 
\end{equation}
\end{lemma}

\begin{proof}
Since $a=c$ and $b=d=0$ when $u=\eta$, the case  $u\rightarrow\eta$ of \eqref{zg} is
$$\Gamma=\frac{a_u-c_u}{b_u}\Bigg|_{u=\eta}. $$
Combining this with \eqref{ler} gives
$$\varepsilon+\frac \Gamma 2=\frac{a_u}{b_u}-\frac{a}{Lb_u}\frac{\lambda_u}{\lambda}\Bigg|_{u=\eta}.$$
Writing  the variables as $(u,\eta)$, we split this expression as $T_1+T_2$, where
\begin{align*}T_1&=\frac{a_u}{b_u}(\eta,\eta)-\frac{a}{Lb_u}(\eta,\eta)\frac{\lambda_u}{\lambda}\left(\eta,\pi/3\right),\\
T_2&=\frac{a}{Lb_u}(\eta,\eta)\left(\frac{\lambda_u}{\lambda}\left(\eta,\pi/3\right)
 -\frac{\lambda_u}{\lambda}\left(\eta,\eta\right)
 \right).
 \end{align*}
 
By \eqref{phf} and \eqref{ese},
$$\frac{\lambda_u}{\lambda}(u,\pi/3)=L\frac{a_u+b_u}{a+b}(u,\eta), $$ 
where $\eta$ on the right-hand side is arbitrary. In particular, 
taking $u=\eta$ and using 
that  $b(\eta,\eta)=0$ gives
$$\frac{\lambda_u}{\lambda}(\eta,\pi/3)=L\frac{a_u+b_u}{a}(\eta,\eta). $$ 
It follows that $T_1=-1$. Hence, the only contribution to \eqref{fni} comes from $T_2$.
To give a non-zero contribution, the derivative must hit the second factor
and we obtain
 $$\frac{\partial T_2}{\partial \eta}\Bigg|_{\eta=\pi/3}=-\frac{a}{Lb_u} \frac{\partial}{\partial \eta}\frac{\lambda_u}{\lambda}\Bigg|_{u=\eta=\pi/3},$$
 which equals the right-hand side of \eqref{fni}. 
\end{proof}

Using also \eqref{se} and \eqref{zze}, we may now rewrite \eqref{feg} as follows.

\begin{proposition}\label{fnp}
The quantity $f_n$ appearing in  \eqref{cf} can be expressed as
$$f_n=\frac{\zeta^2+3}{12L\chi } \frac{\lambda_{u\eta}\lambda-\lambda_u\lambda_\eta}{\lambda^2}\Bigg|_{u=\eta=\pi/3}.$$
\end{proposition}

\subsection{The $TQ$-relation}
\label{tqss}

Let
$$\phi(u)=\lambda(u)\Big|_{\eta=\pi/3}=\theta_1(u|\tau)^L. $$
We have reduced the computation of the correlation functions to the evaluation of
$$
\lambda_{u\eta }\lambda-\lambda_u\lambda_\eta\Big|_{u=\eta=\pi/3}=
\lambda_{u\eta }\phi-\phi_u\lambda_\eta\Big|_{u=\eta=\pi/3}
.$$
To handle the $\eta$-derivatives, we will exploit Baxter's $TQ$-relation.

A $Q$-operator is a family of operators $\mathbf Q=\mathbf Q(u)$ (depending also on $\eta$ and $\tau$) acting on $V^{\otimes L}$ such that $\mathbf Q(u)$ and $\mathbf T(v)$ always commute, and 
$$\mathbf T(u)\mathbf Q(u)=\phi(u-\eta)\mathbf Q(u+2\eta)+\phi(u+\eta)\mathbf Q(u-2\eta). $$ 
It is also natural to assume the quasi-periodicity and reflection symmetry
$$\mathbf Q(u+2\pi)=\mathbf Q(-u)=\mathbf Q(u),\qquad\mathbf  Q(u+2\pi\tau)=e^{-2\ti L(u+\pi\tau)}\mathbf Q(u),$$
which are consistent with \eqref{tqp}.

Baxter gave two different constructions of $Q$-operators for the eight-vertex model. Unfortunately, neither of these work in our situation, that is, for odd $L$ and $\eta=\pi/3$. Baxter's first $Q$-operator \cite{b1} is defined in terms of the inverse of an auxiliary operator $\mathbf Q_R$, which is not invertible for $\eta=\pi/3$. His second definition \cite{b3} only works for even $L$. Different constructions due to Fabricius \cite{f} and Roan \cite{r}, give solutions defined at $\eta=\pi/3$, but as we also need the $\eta$-derivative  at $\eta=\pi/3$, that is not enough for our purposes. One possible way out of this problem was suggested by Bazhanov and Mangazeev \cite{bm}. They propose that the ground-state eigenvector $\Psi$ belongs to a subspace of $V^{\otimes L}$ where $\mathbf Q_R$ is invertible. Restricting  Baxter's first $Q$-operator $\mathbf Q$  to that space, one would then have $\mathbf Q\Psi=Q\Psi$, where the eigenvalue $Q=Q(u)$ satisfies
\begin{subequations}\label{qc}
\begin{gather}\label{lqr} \lambda(u)Q(u)=\phi(u-\eta)Q(u+2\eta)+\phi(u+\eta)Q(u-2\eta),\\
  Q(u+2\pi)=Q(-u)=Q(u),\qquad Q(u+2\pi\tau)=e^{-2\ti L(u+\pi\tau)}Q(u). \end{gather}
\end{subequations}
Our results will be derived from these identities, which still lack a rigorous proof. More formally, we make the following assumption:

\begin{assumption}\label{ass} There exists a non-zero function $Q$ that satisfies \eqref{qc} and is
  analytic for all
  $u\in\mathbb C$, $\operatorname{Im}(\tau)>0$ and $\eta$  in some neighbourhood of $\pi/3$ (that may depend on $u$ and $\tau$). 
  \end{assumption}

\begin{lemma}\label{qtdl}
If \emph{Assumption \ref{ass}} holds, then 
the space of functions satisfying the conditions described there is two-dimensional. It is spanned by
two functions $Q^\pm$ that satisfy
\begin{equation}\label{qpm}Q^\pm(u+\pi)=Q^{\mp}(u),\qquad Q^{\pm}(u+\pi\tau)=\pm e^{-\ti L(u+\pi\tau/2)}Q^{\pm}(u).  \end{equation}
Moreover, writing $Q^\pm=Q^{\pm}(u|\tau)$, these functions can be normalized so that
\begin{equation}\label{qms}Q^\pm(u|\tau+2)=Q^\mp(u|\tau).\end{equation}
\end{lemma}

\begin{proof}
By \eqref{tqp}, the transfer-matrix eigenvalue $\lambda(u)$ obeys
$$ \lambda(u+\pi)=-\lambda(u),\qquad \lambda(u+\pi\tau)=-e^{-\ti L(2u+\pi\tau)}\lambda(u),$$
and we also have
$$ \phi(u+\pi)=-\phi(u),\qquad \phi(u+\pi\tau)=-e^{-\ti L(2u+\pi\tau)}\phi(u).$$
Using these relations, it is easy to check that if $Q(u)$ solves \eqref{qc}, then so do $Q(u+\pi)$ and $e^{\ti L(u+\pi\tau/2)}Q(u+\pi\tau)$.
We can then write any solution as $Q=\tilde Q_++\tilde Q^-$, where
$$\tilde Q^{\pm}(u)=\frac 12\left(Q(u)\pm e^{\ti L(u+\pi\tau/2)}Q(u+\pi\tau)\right).$$
If $Q$ is non-zero, the functions $\tilde Q^\pm$  cannot both vanish identically. 
If $\tilde Q^+\not\equiv 0$, we let $Q^+=\tilde Q^+$ and define
$Q^-(u)=Q^+(u+\pi)$. Otherwise, we let $Q^-=\tilde Q^-$ and $Q^+(u)=Q^-(u+\pi)$.  
It is then clear that $Q^\pm$ satisfy \eqref{qc} and \eqref{qpm}.

We will write $q(u)=Q^+(u)\big|_{\eta=\pi/3}$. This function satisfies
\begin{subequations}\label{qdp}
\begin{gather} \phi(u)q(u)+\phi\left(u+\frac{2\pi}3\right)q\left(u+\frac{2\pi }3\right)+\phi\left(u-\frac{2\pi }3\right)q\left(u-\frac{2\pi}3\right)=0,\\
 q(u+2\pi)=q(-u)=q(u),\qquad q(u+\pi\tau)=e^{-\ti L(u+\pi\tau/2)}q(u). \end{gather}
 \end{subequations}
 It is proved in \cite{rsa} that
 the space of entire functions subject to these conditions is one-dimensional.
(More precisely, it follows from the discussion in 
 \cite[\S 5.4]{rsa}  that this is a consequence of 
 \cite[Thm.\ 2.4]{rsa}.)
  Since a small perturbation cannot increase the dimension of the solution space, 
our assumption that a non-zero $Q$ exists implies  that $Q^+$ remains unique for $\eta$ near $\pi/3$.  
In particular, the function $\tilde Q^+$  is proportional to $Q^+$ and similarly $\tilde Q^-$ is proportional to $Q^-$.
Hence, any solution $Q$ is in the span of $Q^+$ and $Q^-$.

For the final statement, we first note that $Q^-(u|\tau+2)=Q^+(\pi-u|\tau+2)$ satisfies the properties defining $Q^+$, so there exists
$f(\tau)$ (depending also on $\eta$) such that 
$$Q^+(\pi-u|\tau+2)=f(\tau)Q^+(u|\tau).$$ 
In particular, $u=\pi/2$ has the same multiplicity as a zero of both $Q^+(u|\tau+2)$ and $Q^+(u|\tau)$. If $k$ is this multiplicity and $a(\tau)$ is the leading Taylor coefficient of $Q^+(u|\tau)$ at $u=\pi/2$, then 
$$f(\tau)=\lim_{u\rightarrow\pi/2}\frac{Q^+(\pi-u|\tau+2)}{Q^+(u|\tau)}=(-1)^k\frac{a(\tau+2)}{a(\tau)}. $$
It is now easy to check that, if $k$ is even, then the renormalized functions $\tilde Q^\pm=Q^\pm/a$ satisfy $\tilde Q^\pm(u|\tau+2)=\tilde Q^\mp(u|\tau)$. If $k$ is odd, one can instead take  $\tilde Q^\pm=e^{\ti\pi\tau/2}Q^\pm/a$.
\end{proof}

From now on, we fix two functions $Q^\pm$  satisfying all conditions in
Lemma \ref{qtdl}. Writing, as in the proof, $q=Q^+|_{\eta=\pi/3}$, 
we have in particular
\begin{equation}\label{qsm}q(u+\pi|\tau)=q(u|\tau+2).\end{equation}

The functions $Q^\pm$ satisfy the Wronskian relations \cite{bm2}
\begin{subequations}\label{bmr}
\begin{align}\label{qw}Q^+(u-\eta)Q^-(u+\eta)-Q^+(u+\eta)Q^-(u-\eta)&=W\phi(u),\\
\label{qql}Q^+(u-2\eta)Q^-(u+2\eta)-Q^+(u+2\eta)Q^-(u-2\eta)&=W\lambda(u),\end{align}
\end{subequations}
where $W=W(\eta,\tau)$ is  independent of $u$. 
 We introduce the difference operator
 $$(\mathbf Rf)(\tau)=f(\tau)-f(\tau+2). $$
 Using  \eqref{qms},
 we can then write \eqref{bmr} more compactly as
 \begin{subequations}\label{rqw}
  \begin{align}\mathbf R\left(Q^+(u-\eta)Q^-(u+\eta)\right)&=W\phi(u),\\
\mathbf R\left(Q^+(u-2\eta)Q^-(u+2\eta)\right)&=W\lambda(u).\end{align}
\end{subequations}
At $\eta=\pi/3$, these relations are equivalent and simplify to
\begin{equation}\label{qwsc} W\phi(u)=
\mathbf R\left(q\left(u-\frac\pi 3\right)q\left(u-\frac{2\pi}3\right)\right).
\end{equation}

The following lemma expresses the $\eta$-derivatives of the transfer-matrix eigenvalue in terms of derivatives of the function $q$ alone. The expressions follow from a cancellation for which the supersymmetric value $\eta=\pi/3$ is crucial. This simple but important feature was first observed by Stroganov \cite{st,st2}. It enables the computation of the correlation functions from the $TQ$-relation.
\begin{lemma}\label{ldl}
For $\eta=\pi/3$, 
\begin{subequations}
\label{cdw}
\begin{align}
\label{leq}\frac{W \lambda_\eta}3&=
\mathbf R\left( q\left(u-\frac{2\pi}{3}\right)q'\left(u-\frac{\pi}{3}\right)-q'\left(u-\frac{2\pi}{3}\right)q\left(u-\frac{\pi}{3}\right)\right),\\
\label{lue} \frac{W \lambda_{u\eta}}3&=\mathbf R\left(
q\left(u-\frac{2\pi}{3}\right)q''\left(u-\frac{\pi}{3}\right)-q''\left(u-\frac{2\pi}{3}\right)q\left(u-\frac{\pi}{3}\right)
\right).
\end{align}
\end{subequations}
\end{lemma}

Here and below, primes refer to derivatives in the spectral parameter $u$.
\begin{proof}
Differentiating \eqref{rqw} with respect to $\eta$ gives
\begin{align*}W_\eta\phi&=\mathbf R\Big(-Q^+_u(u-\eta)Q^-(u+\eta)+Q^+(u-\eta)Q^-_u(u+\eta)\\
&\quad+Q^+_\eta(u-\eta)Q^-(u+\eta)+Q^+(u-\eta)Q^-_\eta(u+\eta)\Big),\\
W_\eta\lambda+W\lambda_\eta&=\mathbf R\Big(-2Q^+_u(u-2\eta)Q^-(u+2\eta)+2Q^+(u-2\eta)Q^-_u(u+2\eta)\\
&\quad+Q^+_\eta(u-2\eta)Q^-(u+2\eta)+Q^+(u-2\eta)Q^-_\eta(u+2\eta)\Big).
\end{align*}
At $\eta=\pi/3$, we may write the second relation as
\begin{align*}W_\eta\phi+W\lambda_\eta&=\mathbf R
\Big(-2Q^-_u(u+\eta)Q^+(u-\eta)+2Q^-(u+\eta)Q^+_u(u-\eta)\\
&\quad+Q^-_\eta(u+\eta)Q^+(u-\eta)+Q^-(u+\eta)Q^+_\eta(u-\eta)\Big).
\end{align*}
Subtracting  the first relation, all terms involving $\eta$-derivatives cancel and we obtain
\begin{align*}W\lambda_\eta&=3\mathbf R\big(Q^-(u+\eta)Q^+_u(u-\eta)-Q^-_u(u+\eta)Q^+(u-\eta)\big).
\end{align*}
This can be expressed as \eqref{leq}, which gives \eqref{lue} after differentiation.
\end{proof}

Specializing  $u=\pi/3$ in \eqref{qwsc} and its first two $u$-derivatives, as well as in \eqref{cdw},  gives the following useful relations.

\begin{corollary}\label{wrc}
At $u=\eta=\pi/3$, the following identities hold:
\begin{align*}
W\phi&=\mathbf R\left(q(0)q\left(\frac\pi3\right)\right),\\
W\phi_u&=-\mathbf R\left(q(0)q'\left(\frac\pi3\right)\right),\\
 W\phi_{uu}&=\mathbf R\left(q(0)q''\left(\frac\pi3\right)
+q''(0)q\left(\frac\pi3\right)\right),\\
\frac {W \lambda_\eta}3&=
\mathbf R\left(q\left(0\right)q'\left(\frac{\pi}{3}\right)\right),\\
\frac {W\lambda_{u\eta}}3&=\mathbf R\left(q''\left(0\right)
q\left(\frac{\pi}{3}\right)-q\left(0\right)q''\left(\frac{\pi}{3}\right)
\right).
\end{align*}
\end{corollary}

Using these identities in Proposition \ref{fnp} gives the following result.

\begin{corollary}\label{lxc}
We have
\begin{equation}\label{cx}f_n
=-\frac{(\zeta^2+3)}{4L\chi\phi(\pi/3)W}\,\mathbf R(X),
\end{equation}
where
$$X=q(0)\left(q''\left(\frac{\pi}3\right)+\frac{\phi'}\phi\left(\frac{\pi}3\right)q'\left(\frac{\pi}{3}\right)\right)-q''(0)q\left(\frac{\pi} 3\right). $$
\end{corollary}

 \subsection{Differential relations}
 \label{brs}
We will need the following difference-differential equation for $q(u)$, which may have some independent interest. Recall that the Weierstrass function $\wp(u|\tau_1,\tau_2)$ is the unique meromorphic function with period lattice $\mathbb Z\tau_1+\mathbb Z\tau_2$ that has poles only at the lattice points and satisfies
\begin{equation}\label{wpc}
\wp(u|\tau_1,\tau_2)=\frac 1{u^2}+\mathcal O(u^2),\quad u\rightarrow 0. 
\end{equation}

\begin{theorem}\label{ddt}
The $Q$-operator eigenvalues $q(u)$ and $q(u+\pi)$ are related by
\begin{equation}\label{ddr}\left(\frac{\partial^2}{\partial u^2}-V(u)-\alpha\right)\frac{\phi(u)q(u)}{\theta_1(3u|3\tau)^n\,\theta_3\left(\frac{3u}2\big|\frac{3\tau}2\right)}=\beta\frac{\theta_4\left(\frac{3u}2\big|\frac{3\tau}2\right)\phi(u)q(u+\pi)}{\theta_1(3u|3\tau)^n\,\theta_3\left(\frac{3u}2\big|\frac{3\tau}2\right)^2},
 \end{equation}
where 
$$V(u)=n(n+1)\wp\left(u\big|\frac\pi3,\pi\tau\right)+2\wp\left(u+\pi+\frac{\pi\tau}2\Big|\frac{2\pi}3,\pi\tau\right) $$
and the parameters $\alpha$ and $\beta$ are independent of $u$.
\end{theorem}

Before proving Theorem \ref{ddt}, we write it in a form closer to how the classical \eqref{cpvi} and quantum \eqref{qpvi} Painlev\'e VI equations are given in the introduction. We must then scale the periods $(2\pi/3,\pi\tau)$ to $(1,\tau)$, that is, we should write $u=2\pi x/3$ and replace $\tau$ by $2\tau/3$. 
Writing $q(u)=q(u,\tau)$ we find after using some elementary identities for the $\wp$-function \cite[\S 20]{ww} that
$$\psi=\psi(x,\tau)=\frac{\theta_1\left(\frac{2\pi x}{3}\big|\frac{2\tau}{3}\right)^{2n+1}q\left(\frac{2\pi x}{3},\frac{2\tau}{3}\right)}{\theta_1(2\pi x|2\tau)^n\theta_3(\pi x|\tau)} $$
satisfies an equation of the
form
$$\frac 12\,\psi_{xx}-V\psi=(A+B\wp(x+\omega_3|1,\tau))\psi\left(x+\frac 32,\tau\right), $$
where $A$ and $B$ are independent of $x$ and $V$ is the potential \eqref{dtv} with $\beta_0=\beta_1=n(n+1)/2$, $\beta_2=0$ and $\beta_3=1$. Note that 
$\psi$ is $3$-periodic  in $x$ whereas the potential is $1$-periodic.
This can be compared with
the non-stationary Lam\'e equation  \cite{bm}, which can be written
\begin{equation}\label{nsl}\frac 12\,\tilde \psi_{xx}-\tilde V\tilde \psi=2\pi\ti\tilde\psi_\tau, \end{equation}
with
$$\tilde \psi(x,\tau)=\frac{e^{c\tau}\theta_1(\frac{2\pi x}{3}\big|\frac{2\tau}{3})^{2n+1}q\left(\frac{2\pi x}{3},\frac{2\tau}{3}\right)}{\theta_1(2\pi x|2\tau)^n}, $$
$c$ a normalizing constant
and $\tilde V$ the potential \eqref{dtv} with  parameters $\beta_0=\beta_1=n(n+1)/2$, $\beta_2=\beta_3=0$.
Originally, we discovered \eqref{ddr}  as a consistency condition between \eqref{nsl} and 
another differential equation from \cite{r4}, for the function $\Phi$ defined in \eqref{psi}.

To prepare the proof of Theorem \ref{ddt}, let 
$$f_1(u)=\theta_3\left(\frac{3u}2\Big|\frac{3\tau}2\right)q(u),\qquad f_2(u)=f_1(u+\pi)$$
and write $\Omega=g^{-1} (\partial_u^2-V(u))g$,  where
$$g(u)=\frac{\theta_1(u|\tau)^{2n+1}}{\theta_1(3u|3\tau)^n\,\theta_3\left(\frac{3u}2\big|\frac{3\tau}2\right)^2}.  $$
With these definitions at hand, the differential relation \eqref{ddr} takes the form
$\Omega(f_1)=\alpha f_1+\beta f_2$. 
To prove it, we first give an analytic description of the space $\Theta$ spanned by $f_1$ and $f_2$, and then prove that $\Omega(f_1)\in\Theta$. We remark that $\Omega(f_2)\notin\Theta$.

\begin{lemma}\label{tdl}
Let $\Theta$ be the space of all entire functions $f$ that satisfy 
\begin{subequations}
  \label{prr}
  \begin{gather}\label{prra} \phi(u)f(u)+\phi\left(u+\frac{2\pi}3\right)f\left(u+\frac{2\pi }3\right)+\phi\left(u-\frac{2\pi }3\right)f\left(u-\frac{2\pi}3\right)=0,\\
  \label{prrb} f(u+2\pi)=f(-u)=f(u),\qquad f(u+\pi\tau)=e^{-\ti (n+2)(2u+\pi\tau)}f(u). \end{gather}
\end{subequations}
Then, $\Theta$ is spanned by $f_1$ and $f_2$.
\end{lemma}

The fact that $\Theta$ is two-dimensional is also a special case of \cite[Thm.\ 2.4]{rsa}. To make the connection, note that$f\in \Theta$ if and only if the function $(\phi f)(u)/\theta_4(3u|3\tau)$ is in the space denoted $\Theta_n^{(n,-1,-1,n)}$ in \cite{rsa}, where the variables $(z,\tau)$ in \cite{rsa} correspond to our $(u/2\pi,\tau/2)$.

\begin{proof}[Proof of \emph{Lemma \ref{tdl}}]
It is straightforward to check from \eqref{qdp} that $f_1,\,f_2\in \Theta$. Consider the functional $(\Lambda f)(u)=f(\pi+\pi\tau/2)$ defined on $\Theta$. We observe that substituting $u=\pi+\pi\tau/2$ in \eqref{prra} and using \eqref{prrb} gives $$(\phi f) \left(\pi+\frac{\pi\tau}2\right)+2(\phi f) \left(\pi+\frac{\pi\tau}2+\frac{2\pi }3\right)=0.$$ This equality shows that if $f\in \Ker(\Lambda)$, then $f$ vanishes at all zeroes of $\theta_3(3u/2|3\tau/2)$. Hence, $f/\theta_3(3u/2|3\tau/2)$ is an entire function satisfying \eqref{qdp}, so $f$ is proportional to $f_1$.  
   
On the other hand, it follows from \cite[Thm.\ 2.1]{r4} that  $q(\pi\tau/2)\neq 0$, so $f_2\notin\Ker(\Lambda)$. Since both the kernel and image of $\Lambda$ are one-dimensional, $\dim(\Theta)=\dim \, \Ker(\Lambda) + \dim \, \operatorname{Im}(\Lambda)=2$.
\end{proof}

\begin{proof}[Proof of \emph{Theorem \ref{ddt}}]
It remains to verify that 
$h=\Omega f_1\in\Theta$. 
 It is straight-forward to check that $f$ satisfies \eqref{prr} if and only if
 \begin{gather}\label{gfca} (gf)(u)+(gf)\left(u+\frac{2\pi}3\right)+(gf)\left(u-\frac{2\pi }3\right)=0,\\
 \nonumber (gf)(u+2\pi)=(gf)(u),\qquad (gf)(u+\pi\tau)=(gf)(-u)=(-1)^{n+1}(gf)(u). \end{gather}
 Since the operator
$\partial_u^2-V(u)$ preserves these conditions, we only need to check that $h$ is an entire function.

The only possible poles of $h$ are at the zeroes and poles of $g$ (which include the poles of the potential $V$). Modulo translations by $2\pi\mathbb Z+\pi\tau\mathbb Z$ and the reflection $u\mapsto -u$, there are six such points, namely,
\begin{equation}\label{shp}0,\quad \frac\pi 3,\quad\frac{2\pi}3,\quad\pi, \quad \frac{\pi}{3}+\frac{\pi\tau}2,\quad \pi+\frac{\pi\tau}2. \end{equation}

It is straightforward to check that $h$ is regular at the points \eqref{shp} and we only provide details for the point $u_0=2\pi/3$. It is a zero of  $\theta_1(3u|3\tau)$, so $(gf_1)(u_0+u)=\mathcal O(u^{-n})$. Moreover, it follows from \eqref{gfca} that
$$(gf_1)(u_0+u)+(-1)^{n+1}(gf_1)(u_0-u)=-(gf_1)(u). $$
Since $u=0$ is a zero of both $\theta_1(3u|3\tau)$ and $\theta_1(u|\tau)$, $(gf_1)(u)=\mathcal O(u^{n+1})$.
Hence, 
$$(gf_1)(u_0+u) =Au^{-n}+\mathcal O(u^{2-n})+Bu^{n+1}+\mathcal O(u^{n+3}). $$
If $n\geq 1$, this simplifies to
$$(gf_1)(u_0+u) =Au^{-n}+\mathcal O(u^{2-n}). $$
Acting with 
\begin{equation}\label{pta} \frac{\partial^2}{\partial u^2}-V(u_0+u)=\frac{\partial^2}{\partial u^2}-\frac{n(n+1)}{u^2}+\mathcal O(1),\end{equation}
the coefficient of $u^{-n-2}$ cancels, so the result is  $\mathcal O(u^{-n})$. If $n=0$, we have instead $$(gf_1)(u_0+u) =A+Bu+\mathcal O(u^{2})=\mathcal O(1). $$
In this case \eqref{pta} simplifies to $\partial_u^2+\mathcal O(1)$, so the result is again $\mathcal O(1)$.  In either case, dividing by $g$ leads to a function regular at $u_0$. 
\end{proof}

Let $\Phi$ denote the alternant
\begin{align}\nonumber\Phi(u,v)&=f_1(u)f_2(v)-f_2(u)f_1(v)\\
\nonumber&=\theta_3\left(\frac{3u}2\Big|\frac{3\tau}2\right)\theta_4\left(\frac{3v}2\Big|\frac{3\tau}2\right)q(u)q(v+\pi)\\
\label{psi}&\quad -
\theta_4\left(\frac{3u}2\Big|\frac{3\tau}2\right)\theta_3\left(\frac{3v}2\Big|\frac{3\tau}2\right)q(u+\pi)q(v). \end{align}
We note some properties of this function. 
Indicating also the $\tau$-dependence, it follows from \eqref{qsm} that
\begin{equation}\label{psm}\Phi(u,v|\tau+2)=-\Phi(u,v|\tau). \end{equation}
The Wronskian relation \eqref{qw} is equivalent to
 $$\Phi\left(u,u+\frac{2\pi}3\right)=\theta_3\left(\frac{3u}2\Big|\frac{3\tau}2\right)\theta_4\left(\frac{3u}2\Big|\frac{3\tau}2\right) \phi\left(u+\frac\pi3\right)W.$$
 Inserting $u=\pi/3$ and using the theta function identity \eqref{jtv}
  gives
 \begin{equation}\label{waa}W=-\frac{\theta_2\left(0\big|\frac{3\tau}2\right)\Phi\left(\pi,\frac{\pi}3\right)}{\theta_1'\left(0|\frac{3\tau}2\right) \phi\left(\frac\pi3\right)}. \end{equation}

The following consequence of Theorem \ref{ddt} is a key result for 
our purposes.

\begin{corollary}\label{qfc}
The $Q$-operator eigenvalue $q$ satisfies
\begin{multline}\label{qqeps}
(2n+3)q''(0)q\left(\frac\pi 3\right)+(2n-1)q(0)\frac{(q\phi)''\left(\frac\pi 3\right)}{\phi\left(\frac\pi 3\right)}
+2(2n+1)Eq(0)q\left(\frac\pi 3\right)
\\
=3\ti\frac{q'\left(\pi+\frac{\pi\tau}2\right)}{q\left(\frac{\pi\tau}2\right)}\,\Phi\left(0,\frac\pi3\right),
\end{multline}
where $E$ is defined by the expansion
\begin{equation}\label{ed}\frac{\theta_1(u|\tau)^{2n+3}}{\theta_1(3u|3\tau)^{2n}\theta_4(3u|3\tau)}
=E_0u^{3}+E_0Eu^5+\mathcal O(u^7). \end{equation}
\end{corollary}

 \begin{proof}
We will consider \eqref{ddr} near $u=u_0$, where
 $u_0= \pi+\pi\tau/2$, $0$ and $\pi/3$. In the first case, $u_0$ is a zero of $q(u)$ and of $\theta_3(3u/2|3\tau/2)$. The only contribution to the  leading term on the left-hand side comes from the singular part of the potential, which is $2(u-u_0)^{-2}$. 
Thus, 
$$
\beta=\lim_{u\rightarrow u_0}
-\frac 2{(u-u_0)^2}\frac{\theta_3\left(\frac{3u}2\big|\frac{3\tau}2\right)q(u)}{\theta_4\left(\frac{3u}2\big|\frac{3\tau}2\right)q(u+\pi)}.
$$
We have
$$\frac{\theta_3\left(\frac{3u}2\big|\frac{3\tau}2\right)}{\theta_4\left(\frac{3u}2\big|\frac{3\tau}2\right)}=\ti \frac{\theta_1\big(\frac{3(u-u_0)}2\big|\frac{3\tau}2\big)}{\theta_2\big(\frac{3(u-u_0)}2\big|\frac{3\tau}2\big)}=\frac {3\ti}2\frac{\theta_1'\left(0\big|\frac{3\tau}2\right)}{\theta_2\left(0\big|\frac{3\tau}2\right)}(u-u_0)+\mathcal O((u-u_0)^3),$$
$$\frac{q(u)}{q(u+\pi)}=\frac{q'(\pi+\frac{\pi\tau}2)}{q(\frac{\pi\tau}2)}(u-u_0)+ \mathcal O((u-u_0)^3).$$
Using also \eqref{jtv}
gives 
\begin{equation}\label{beq}\beta=-3\ti
 \,\theta_3\left(0\Big|\frac{3\tau}2\right)\theta_4\left(0\Big|\frac{3\tau}2\right)
 \frac{q'\left(\pi+\frac{\pi\tau}2\right)}{q\left(\frac{\pi\tau}2\right)}.
  \end{equation}

Turning to the case $u_0=0$,
we define $A_0$ and $A$ by
$$\frac{\phi(u)}{\theta_1(3u|3\tau)^n\theta_3\left(\frac{3u}2\big|\frac{3\tau}2\right)}=A_0u^{n+1}+A_0Au^{n+3}+\mathcal O(u^{n+5}),$$
so that
 $$\frac{\phi(u)q(u)}{\theta_1(3u|3\tau)^n\theta_3\left(\frac{3u}2\big|\frac{3\tau}2\right)}=A_0q(0)u^{n+1}+A_0\left(\frac{q''(0)}2+Aq(0)\right) u^{n+3}+\mathcal O(u^{n+5}).$$
 The potential has the form
 $$V(u)=\frac{n(n+1)}{u^2}+2C +\mathcal O(u^2), $$
 where
 $$C= \wp\left(\frac\pi3+\frac{\pi\tau}2\Big|\frac{2\pi}3,\pi\tau\right).$$
 Inserting these expansions 
 into \eqref{ddr} and picking out the coefficient of $u^{n+1}$ gives
 \begin{subequations}\label{ab}
\begin{equation}  q(0)\alpha+\frac{\theta_4\left(0\big|\frac{3\tau}2\right)}{\theta_3\left(0\big|\frac{3\tau}2\right)}\,q(\pi)\beta=
 (2n+3)\big(q''(0)+2Aq(0)\big) -2Cq(0).\end{equation}
 In the same way, the expansion of \eqref{ddr} near $u_0=\pi/3$ gives
\begin{multline}  q\left(\frac\pi 3\right)\alpha+\frac{\theta_3\left(0\big|\frac{3\tau}2\right)}{\theta_4\left(0\big|\frac{3\tau}2\right)}\,q\left(\frac{2\pi }3\right)\beta\\
=
 -(2n-1) \left(\frac{(q\phi)''\left(\frac\pi3\right)}{\phi\left( \frac\pi3\right)}+2Bq\left(\frac\pi 3\right)\right) -2 Dq\left(\frac\pi 3\right),\end{multline}
 \end{subequations}
 where 
$$\frac{1}{\theta_1(3u|3\tau)^n\theta_4\left(\frac{3u}2\big|\frac{3\tau}2\right)} =B_0u^{-n}+B_0B u^{2-n}+\mathcal O(u^{4-n}) $$
and
$$D= \wp\left(\frac{\pi\tau}2\Big|\frac{2\pi}3,\pi\tau\right).$$
 
We now
eliminate $\alpha$ from the pair of equations \eqref{ab} and insert the expression \eqref{beq} for $\beta$. This leads to \eqref{qqeps}, with
\begin{equation}\label{abcde}(2n+1)E=(2n+3)A+(2n-1)B-C+D. \end{equation}
It remains to show that this agrees with \eqref{ed}. 

It follows from the elementary theory of elliptic functions that
$$\wp\left(u\big|\frac{2\pi}3,\pi\tau\right)=P\frac{\theta_3\left(\frac{3u}2\big|\frac{3\tau}2\right)^2}{\theta_1\left(\frac{3u}2\big|\frac{3\tau}2\right)^2}+Q, $$
with $P$ and $Q$ independent of $u$. Substituting $u=\pi/3+\pi\tau/2$, the first term on the right vanishes and we see that $Q=C$. Using
\eqref{wpc} then gives
$$ \frac{\theta_1\left(\frac{3u}2\big|\frac{3\tau}2\right)^2}{\theta_3\left(\frac{3u}2\big|\frac{3\tau}2\right)^2}=Pu^2+PCu^4+\mathcal O(u^6).$$
Replacing $\theta_3$ by $\theta_4$ gives a similar identity for $D$. Thus, the right-hand side of \eqref{abcde} is the quotient of the subleading and the leading coefficients of
\begin{multline*}\left(\frac{\theta_1(u|\tau)^{2n+1}}{\theta_1(3u|3\tau)^n\theta_3\left(\frac{3u}2\big|\frac{3\tau}2\right)}\right)^{2n+3}\left(\frac{1}{\theta_1(3u|3\tau)^n\theta_4\left(\frac{3u}2\big|\frac{3\tau}2\right)}\right)^{2n-1}\frac{\theta_3\left(\frac{3u}2\big|\frac{3\tau}2\right)^2}{\theta_4\left(\frac{3u}2\big|\frac{3\tau}2\right)^2}\\
=\left(\frac{\theta_1(u|\tau)^{2n+3}}{\theta_1(3u|3\tau)^{2n}\theta_3\left(\frac{3u}2\big|\frac{3\tau}2\right)\theta_4\left(\frac{3u}2\big|\frac{3\tau}2\right)} \right)^{2n+1}.
\end{multline*}
Using also \eqref{t4d}, we find that this quotient equals $(2n+1)E$. 
\end{proof}

 \subsection{Splitting of the function $X$}

We now return to the expression \eqref{cx}. We split the function $X$ into two parts, 
which will eventually correspond to the infinite-lattice limit and the finite-size correction of the correlation functions. More precisely, we write $X=X_1+X_2$, where
\begin{align*}
X_1&=
\frac{2n+1}{2}\left(q''\left(\frac{\pi}{3}\right)q(0)+q\left(\frac{\pi}{3}\right)q''(0)
  \right) +2n\frac{\phi'}{\phi}\left(\frac{\pi}3\right)q'\left(\frac{\pi}{3}\right)q(0)\\
  &\quad+\left(\frac{2n-1}{2}\frac{\phi''}{\phi}\left(\frac{\pi}3\right)+(2n+1)E\right) q\left(\frac{\pi}{3}\right)q(0),\\
X_2&=
-\frac{2n+3}2q''(0)q\left(\frac\pi 3\right)-\frac{2n-1}2q(0)\frac{(q\phi)''\left(\frac\pi 3\right)}{\phi\left(\frac\pi 3\right)}
-(2n+1)Eq(0)q\left(\frac\pi 3\right)\\
&=-\frac {3\ti}2\,\frac{q'\left(\pi+\frac{\pi\tau}2\right)}{q\left(\frac{\pi\tau}2\right)}\,\Phi\left(0,\frac\pi3\right),
 \end{align*}  
 where the final expression follows from Corollary \ref{qfc}.
 
Using Corollary~\ref{wrc}, one finds that
$$\mathbf R(X_1)
=W\phi\left(2n\frac{\phi''}{\phi}-2n\left(\frac{\phi'}{\phi}\right)^2+(2n+1)E\right)\bigg|_{u=\pi/3}.$$ 
By Lemma \ref{til}, this can be simplified further to
$$\mathbf R(X_1)=-(2n+1)\phi\left(\frac\pi 3\right)\frac{\chi(\gamma^2-3)W}{(\gamma+1)^2},  $$
where $\chi$ is as in \eqref{chi} and where  as before $\gamma=(\zeta+3)/(\zeta-1)$. Using \eqref{qsm} and \eqref{psm}, it is straightforward to check that $\mathbf R(X_2)=2X_2$. Hence,
$$\mathbf R(X)=-(2n+1)\phi\left(\frac\pi 3\right)\frac{\chi(\gamma^2-3)W}{(\gamma+1)^2}-
3\ti\frac{q'\left(\pi+\frac{\pi\tau}2\right)}{q\left(\frac{\pi\tau}2\right)}\,\Phi\left(0,\frac\pi3\right)
.$$
Inserting this expression and \eqref{waa} in Corollary \ref{lxc} gives the following result.

\begin{lemma}\label{fsl}
The function $f_n$ can be expressed as
\begin{equation}\label{cxb}
f_n=\frac{(\gamma^2-3)(\gamma^2+3)}{(\gamma^2-1)^2}
-\frac{3\ti(\gamma^2+3)}{(2n+1)\chi(\gamma-1)^2}
\frac{\theta_1'\left(0\big|\frac{3\tau}2\right)}{\theta_2\left(0|\frac{3\tau}2\right) }
\frac{\Phi\left(0,\frac{\pi}3\right)}{\Phi\left(\pi,\frac{\pi}3\right)}\frac{q'\left(\pi+\frac{\pi\tau}2\right)}{q\left(\frac{\pi\tau}2\right)}
.
  \end{equation}
\end{lemma}

\subsection{Proof of Theorem \ref{mt}}

To complete the proof of Theorem \ref{mt}, we need to express the last term in \eqref{cxb} in terms of the polynomials $s_n$ and $\bar s_n$.
 To keep close to the conventions of \cite{r4}, we introduce the variables
$$x=x(u|\tau)=-\frac{\theta_3\big(\frac\pi3\big|\frac\tau2\big)^2\theta_1\big(\frac{3u}2\big|\frac{3\tau}2\big)\theta_4\big(\frac u2\big|\frac\tau2\big)}{\theta_2\big(\frac\pi3\big|\frac\tau2\big)^2\theta_1\big(\frac{u}2\big|\frac{\tau}2\big)\theta_4\big(\frac {3u}2\big|\frac{3\tau}2\big)},$$
\begin{equation}\label{ztau}z=z(\tau)=-\frac{\theta_2\big(0\big|\frac\tau2\big)\theta_3\big(\frac \pi3\big|\frac\tau2\big)}{\theta_3\big( 0\big|\frac\tau2\big)\theta_2\big(\frac \pi3\big|\frac\tau2\big)}. \end{equation}
Then, our variables $(u,\tau,x,z)$ correspond to  $(2\pi z,2\tau,x,\zeta)$ from \cite{r4}. See \eqref{gz} for the relation between the modular functions  $z$ and $\gamma$.

Let $k_0,\,k_1,\,k_2,\,k_3$ and $n$ be integers, such that
$m=2n-\sum_j k_j\geq 0$. In \cite{r4}, the second author introduced a family of rational functions 
$T_n^{(k_0,k_1,k_2,k_3)}(x_1,\dots,x_m;z)$, which are symmetric in the variables $x_j$. They are defined by  explicit determinant formulas that we do not repeat here. When $m=0$, the variables $x_j$ are absent. We denote this special case 
\begin{equation}\label{lt}t^{(k_0,k_1,k_2,k_3)}(z)=T_n^{(k_0,k_1,k_2,k_3)}(-;z),\qquad k_0+k_1+k_2+k_3=2n. \end{equation}

It follows from \cite[Prop.\ 2.17]{rsa} that we obtain a solution to
 \eqref{qdp} as
$$q(u)=q(u|\tau)=C_1 h(u)T_n^{(n,-1,0,n)}(x;z),$$
where $C_1$ is independent of $u$ and 
$$h(u)=\theta_3\left(\frac u2\Big|\frac\tau2\right)\left(\frac{\theta_4\left(\frac{3u}2\big|\frac{3\tau}2\right)}{\theta_4\left(\frac{u}2\big|\frac{\tau}2\right)}\right)^n.$$
We choose $C_1$ so that   \eqref{qsm} is satisfied. (One can show that $C_1=(z+1)^{-3n(n-1)/2}$ works, but we will not need that fact.)
 
 Denoting the half-periods in $2\pi\mathbb Z+\pi\tau\mathbb Z$ by
$$\gamma_0=0,\qquad \gamma_1=\frac{\pi\tau}2,\qquad 
\gamma_2=\pi+\frac{\pi\tau}2,\qquad \gamma_3=\pi, $$
we have, for each $0\leq j\leq 3$,
\begin{equation}\label{ts}T_n^{(k_0,k_1,k_2,k_3)}(x_1,\dots,x_{m-1},x(\gamma_j|\tau);z)
=T_n^{(k_0,k_1,k_2,k_3)+e_j}(x_1,\dots,x_{m-1};z),
 \end{equation}
where $e_0,\dots,e_3$ are the canonical unit vectors of $\mathbb Z^4$. It follows that
\begin{subequations}\label{qpst}
\begin{equation}\label{qg4}q\left(\frac{\pi\tau}2\right)=C_1{h\left(\frac{\pi\tau}{2}\right)}\, t^{(n,0,0,n)}(z)\end{equation}
and, since $h(\gamma_2)=0$,
\begin{equation} q'\left(\pi+\frac{\pi\tau}2\right)=C_1{h'\left(\pi+\frac{\pi\tau}2\right)}\, t^{(n,-1,1,n)}(z). \end{equation}

Next, we consider the function $f(u)=\Phi(u,\pi/3)$. 
It is a solution of \eqref{prr} that satisfies the additional restriction
$f(\pi/3)=0$. It follows from  \cite[Prop.\ 2.17]{rsa}
that 
$$f(u)=C_2k(u)T_{n-1}^{(n,-1,-1,n-1)}(x;z), $$
where $C_2$ is independent of $u$ and
\begin{equation}\label{ku}k(u)=\frac{\theta_2\left(\frac {3u}2\big|\frac{3\tau}2\right)^2\theta_3\left(\frac {3u}2\big|\frac{3\tau}2\right)}{\theta_2\left(\frac u2\big|\frac\tau2\right)^2\theta_3\left(\frac u2\big|\frac\tau2\right)}\left(\frac{\theta_4\left(\frac{3u}2\big|\frac{3\tau}2\right)}{\theta_4\left(\frac{u}2\big|\frac{\tau}2\right)}\right)^{n-1}.
\end{equation}
Applying again \eqref{ts} gives
\begin{align}
\Phi\left(0,\frac\pi 3\right)&=C_2 k(0)t^{(n+1,-1,-1,n-1)}(z),\\
\Phi\left(\pi,\frac\pi 3\right)&=C_2k(\pi)t^{(n,-1,-1,n)}(z).
\end{align}
\end{subequations}

Inserting the expressions \eqref{qpst} in Lemma \ref{fsl}, the factors $C_j$ cancel and we obtain 
\begin{multline}\label{fhkt}f_n=\frac{(\gamma^2-3)(\gamma^2+3)}{(\gamma^2-1)^2}\\
-\frac{3\ti(\gamma^2+3)}{(2n+1)\chi(\gamma-1)^2}\frac{\theta_1'\left(0\big|\frac{3\tau}2\right)}{\theta_2\left(0|\frac{3\tau}2\right) }
\frac{h'\left(\pi+\frac{\pi\tau}2\right)}{h\left(\frac{\pi\tau}2\right)}\frac{k(0)}{k(\pi)}\frac{t^{(n,-1,1,n)}t^{(n+1,-1,-1,n-1)}}{t^{(n,0,0,n)}t^{(n,-1,-1,n)}}.
\end{multline}
We prove in Lemma \ref{chkl} that
\begin{equation}\label{hki}\frac 1{\chi}\frac{\theta_1'\left(0\big|\frac{3\tau}2\right)}{\theta_2\left(0|\frac{3\tau}2\right) }
\frac{h'\left(\pi+\frac{\pi\tau}2\right)}{h\left(\frac{\pi\tau}2\right)}\frac{k(0)}{k(\pi)}=\frac{(-1)^{n+1}2\ti(z+1)}{3(\gamma+1)^2(z-1)(2z+1)^{n-1}}.
\end{equation}
Moreover, it follows from \eqref{stfc} that 
$$\frac{t^{(n,-1,1,n)}t^{(n+1,-1,-1,n-1)}}{t^{(n,0,0,n)}t^{(n,-1,-1,n)}}
 =\frac{(-1)^{n+1}(z-1)^2(2z+1)^n}{(2n+1)(z+1)^2}\frac{\bar s_n(\gamma^{-2})\bar s_{-n-1}(\gamma^{-2})}{s_n(\gamma^{-2})s_{-n-1}(\gamma^{-2})}.
$$
 Using these identities and \eqref{gz} in \eqref{fhkt}, we finally obtain
 \begin{equation}\label{fng} f_n=\frac{(\gamma^2+3)(\gamma^2-3)}{(\gamma^2-1)^2}-\frac{2\gamma^2(\gamma^2+3)}{(2n+1)^2(\gamma^2-1)^2}\frac{\bar s_n(\gamma^{-2})\bar s_{-n-1}(\gamma^{-2})}{s_n(\gamma^{-2})s_{-n-1}(\gamma^{-2})}.
 \end{equation}
 Since we already know that $f_n$ is invariant under interchanging $\gamma$ and $\zeta$, this proves Theorem \ref{mt}.

\section{Connection to Painlev\'e VI}
 
In this section, we prove Theorem \ref{fpqp}. We first briefly review some relevant technicalities. A rational solution of  Painlev\'e VI  can be identified with a homomorphism of differential fields $\mathbb C(q,p,t)\rightarrow \mathbb C(t)$, where $\mathbb C(q,p,t)$ is equipped with the derivation 
$$\delta=\frac{\partial H}{\partial p}\frac{\partial}{\partial q}-\frac{\partial H}{\partial q}\frac{\partial}{\partial p}+t(t-1)\frac{\partial}{\partial t}$$
and $\mathbb C(t)$ with $t(t-1) d/dt$. 
To consider more general solutions algebraically, one needs to work with  field extensions. 
 We follow the approach of \cite{r4}, which uses a differential field $\mathcal F$ generated by $q$, $p$, $t$, $u$, $v$, 
$\alpha_0,\dots,\alpha_4$ and $\tau_0,\dots,\tau_4$. These are subject to the relations
\eqref{ac} as well as
$$u^2v^4=t,\qquad u^4v^2=1-t.$$
Thus,   $u$ and $v$ represent choices of the roots
  $t^{-1/6}(1-t)^{1/3}$ and $t^{1/3}(1-t)^{-1/6}$, respectively. 
The elements $\tau_j$ are abstract tau functions, which represent inverse logarithmic derivatives of modified versions of the Hamiltonian. 
We can then identify the special solution \eqref{sps} with 
 a differential homomorphism
$\mathbf X:\,\mathcal F\rightarrow \mathcal M$, where
 $\mathcal M$ is a field of modular functions. It satisfies
$$\mathbf X(q)=\frac{s(s+2)}{2s+1},\qquad \mathbf X(t)=\frac{s(s+2)^3}{(2s+1)^3}, $$
where
$s=s(\tau)$ is now an element of $\mathcal M$. In \cite{rsc} it was chosen as $s=z$,
with $z$ as  in \eqref{ztau}. To simplify the formulation of Theorem \ref{fpqp} we take here instead 
 $s=-z-1$, which corresponds to making a further 
modular transformation of the variable $\tau$. 

A B\"acklund transformation can be viewed algebraically as
an automorphism of a differential field.  One can define
such transformations  $T_1,\dots,T_4$ that generate an action of
the lattice $\mathbb Z^4$ on $\mathcal F$.  The transformation \eqref{tb} corresponds to $T=T_2^{-1}T_3$. We will write
$$\tau_{l_1l_2l_3l_4}=T_1^{l_1}T_2^{l_2}T_3^{l_3}T_4^{l_4} \tau_0.$$
A  main result of \cite{r4} is that the lattice of modular tau functions $\mathbf X(\tau_{l_1l_2l_3l_4})$
can be identified with the lattice \eqref{lt} of rational functions $t^{(k_0,k_1,k_2,k_3)}$.

To prove Theorem \ref{fpqp}, we first note that
substituting $s=-z-1$ in \eqref{gz} gives 
   $$\gamma^2=\frac{(s+2)(2s+1)}{s}. $$
   As we saw in \eqref{zgds}, this leads to the  relation
   \eqref{zx} between the parameters $Z$ and $s$.

By \eqref{sst} and
  \eqref{smt},
 $$\frac{\bar s_n(\gamma^{-2})\bar s_{-n-1}(\gamma^{-2})}{ s_n(\gamma^{-2}) s_{-n-1}(\gamma^{-2})}
 =\frac{s}{(s+2)(2s+1)^2}\frac{t^{(n,n,1,-1)}(s)t^{(n,n,-2,0)}(s)}{t^{(n,n,0,0)}(s)t^{(n,n,-1,-1)}(s)}.
 $$
  It follows from \cite[Thm.\ 4.2]{r4} (with the variable $\zeta$ there replaced by $-s-1$) that 
  \begin{align*}\frac{t^{(n,n,1,-1)}(s)t^{(n,n,-2,0)}(s)}{t^{(n,n,0,0)}(s)t^{(n,n,-1,-1)}(s)}&=
  -\frac{4(s+2)^3(2s+1)}{s}
     \mathbf X\left(\frac 1{v^2}\frac{\tau_{-1,-n,n,1}\tau_{1,-n,n+1,-1}}{\tau_{0,-n,n,0}\tau_{0,-n,n+1,0}}\right)\\
  &=-\frac{4(s+2)^3(2s+1)}{s}
     \mathbf X T^n\left(\frac 1{v^2}\frac{\tau_{-1,0,0,1}\tau_{1,0,1,-1}}{\tau_{0,0,0,0}\tau_{0,0,1,0}}\right). \end{align*}
    We also used that,
     by \cite[Lemma 2.3]{rsc},   $Tv=v$. By definition,
     $\tau_{0,0,0,0}=\tau_0$ and one can check  check that
 $ \tau_{0,0,1,0}=\tau_3$,
  \begin{align*}
  \tau_{-1,0,0,1}&=-\frac{\ti \tau_0\tau_4 v}{\tau_1}\,(p(q-1)+\alpha_1+\alpha_2),\\
  \tau_{1,0,1,-1}&=-\frac{\ti \tau_1\tau_3v}{\tau_4t}(pq(q-t)+\alpha_2 q+\alpha_4 t).
  \end{align*}
 These expressions allow us to evaluate
 \begin{equation}\label{the}\frac 1{v^2}\frac{\tau_{-1,0,0,1}\tau_{1,0,1,-1}}{\tau_{0,0,0,0}\tau_{0,0,1,0}} =-
     \frac{(p(q-1)+\alpha_1+\alpha_2)(pq(q-t)+\alpha_2 q+\alpha_4 t)}{t}.
   \end{equation}
   We can then rewrite
     \eqref{fng} as
  \begin{align*}
 f_n&=\frac{(s^2+s+1)(s^2+4s+1)}{(s+1)^4}-\frac{4(2s+1)^3(s^2+4s+1)}{(2n+1)^2s(s+1)^4}\\
    &\quad\times \mathbf X T ^n(p(q-1)+\alpha_1+\alpha_2)(pq(q-t)+\alpha_2 q+\alpha_4 t).
  \end{align*}
  Theorem \ref{fpqp} now follows from comparing
  \begin{multline*}
  \mathbf X T ^n(p(q-1)+\alpha_1+\alpha_2)(pq(q-t)+\alpha_2 q+\alpha_4 t)\\
=\left(p_n(q_n-1)+n+\frac 12\right)  \left(p_n(q_n-t)+n+\frac 12\right) q_n 
  \end{multline*}
  with the observation that the Hamiltonian with the parameters \eqref{shpp} factors as
  \begin{equation}\label{hpf}H'+\frac{(2n+1)^2}{4}\,t=
  \left(p(q-1)+n+\frac 12\right)  \left(p(q-t)+n+\frac 12\right)q.
  \end{equation} 
 
 Factorizations such as
 \eqref{hpf} are related to the existence of classical solutions to Painlev\'e VI. In the case at hand, 
 consider solutions such that the first factor in \eqref{hpf} vanishes. The system
 \eqref{ph} then reduces to the single equation $tq'=(n+1/2)q$ and we recover the elementary solutions $q=C t^{n+1/2}$.  
 More generally,  the expression \eqref{the} appears in the factorization
 $$H+(\alpha_0-1)\alpha_3 t=(p(q-1)+\alpha_1+\alpha_2)(pq(q-t)+\alpha_2 q+\alpha_4 t), $$
 which holds for  $\alpha_1+\alpha_2+\alpha_3=0$. In this case, solutions such that the first factor vanishes can be expressed in terms of Gauss' hypergeometric function \cite{o}.

\appendix

\section{Theta function identities}
\label{tfs}

In this appendix we collect some useful theta function identities.
 Fixing $\tau$ in the upper half-plane and  $p=e^{\ti\pi\tau}$,
 the Jacobi  theta functions are defined by the Fourier series
   \begin{subequations}\label{jt}
\begin{align}
\theta_1(u|\tau)&=2\sum_{n=0}^\infty(-1)^np^{(n+1/2)^2}\sin((2n+1)u),\\
\theta_2(u|\tau)&=2\sum_{n=0}^\infty p^{(n+1/2)^2}\cos((2n+1)u),\\
\theta_3(u|\tau)&=1+2\sum_{n=1}^\infty p^{n^2}\cos (2nu),\\
\theta_4(u|\tau)&=1+2\sum_{n=1}^\infty(-1)^np^{n^2}\cos (2nu)
\end{align}
\end{subequations}
or, equivalently, by the product expansions
\begin{subequations}\label{tpe}
 \begin{align}
\label{tpea} \theta_1(u|\tau)&=\ti e^{\ti\pi\tau/4-\ti u}(p^2,e^{2\ti u},p^2e^{-2\ti u};p^2),\\
 \theta_2(u|\tau)&= e^{\ti\pi\tau/4-\ti u}(p^2,-e^{2\ti u},-p^2e^{-2\ti u};p^2),\\
  \theta_3(u|\tau)&=(p^2,-pe^{2\ti u},-pe^{-2\ti u};p^2),\\
 \theta_4(u|\tau)&=(p^2,pe^{2\ti u},pe^{-2\ti u};p^2),
 \end{align}
 \end{subequations}
where
$$(a_1,\dots,a_n;p)_\infty=\prod_{j=0}^\infty(1-a_1p^j)\dotsm(1-a_np^j).$$ Throughout, we use primes, such as in $\theta_1'(u|\tau)$, to indicate the $u$-derivatives of the theta functions.

Elementary quasi-periodicity and reflection relations, such as
$$\theta_1(u+\pi|\tau)=\theta_1(-u|\tau)=-\theta_1(u|\tau),\qquad \theta_1(u+\pi\tau|\tau)=-e^{-\ti(2u+\pi\tau)}\theta_1(u|\tau) $$
will be taken for granted and used without comment. We will also need the modular transformation \cite[\S 21.51]{ww}
\begin{equation}\label{tmt}\theta_4(u/\tau|-1/\tau)=(\tau/\ti)^{1/2}e^{\ti u^2/\pi\tau}\theta_2(u|\tau).
\end{equation}
The  identities
\begin{align}
\label{jtv}\theta_1'(0|\tau)&=\theta_2(0|\tau)\theta_3(0|\tau)\theta_4(0|\tau), \\
\label{t1d}
\theta_4(0|2\tau)\theta_1(2u|2\tau)&=\theta_1(u|\tau)\theta_2(u|\tau),\\
\label{t4d}
\theta_4(0|2\tau)\theta_4(2u|2\tau)&=\theta_3(u|\tau)\theta_4(u|\tau),\\
\label{tde}\theta_2(0|\tau)\theta_1(u|\tau)&=2\,\theta_1(u|2\tau)\theta_4(u|2\tau)
,\\
\label{tdf}\theta_2(0|\tau)\theta_2(u|\tau)&=2\,\theta_2(u|2\tau)\theta_3(u|2\tau),\\
\label{tdpe}
{\theta_1(2u|\tau)}&=2\frac{\theta_1(u|\tau)\theta_2(u|\tau)\theta_3(u|\tau)\theta_4(u|\tau)}{\theta_2(0|\tau)\theta_3(0|\tau)\theta_4(0|\tau)},
\end{align}
\begin{equation}\label{tripl}
\theta_j(3u|3\tau)=\frac{(p^6;p^6)_\infty}{(p^2;p^2)_\infty^3}\,\theta_j(u|\tau)\theta_j\left(\frac\pi 3+u\big|\tau\right)
\theta_j\left(\frac\pi 3-u\big|\tau\right),\quad 1\leq j\leq 4,
\end{equation}
are all easy consequences of  \eqref{tpe}.
Specializing $u=\pi/3$ in \eqref{tdpe} gives
\begin{equation}\label{tdps}\frac{\theta_2(\pi/3|\tau)\theta_3(\pi/3|\tau)
\theta_4(\pi/3|\tau)}{\theta_2(0|\tau)\theta_3(0|\tau)\theta_4(0|\tau)}=\frac 12. \end{equation}

Many identities between the Jacobi theta functions follow from the Weierstrass identity \cite[\S 20.51, Ex. 5, and \S 21.43]{ww}:
\begin{align}
  \label{wsi}
  &\theta_1(x-y|\tau)\theta_1(x+y|\tau)\theta_1(u-v|\tau)\theta_1(u+v|\tau)\\
  -\,&\theta_1(x-u|\tau)\theta_1(x+u|\tau)\theta_1(y-v|\tau)\theta_1(y+v|\tau) \nonumber\\
  +\,&
  \theta_1(x-v|\tau)\theta_1(x+v|\tau)\theta_1(u-y|\tau)\theta_1(u+y|\tau)=0.
  \nonumber
\end{align}
For instance, specializing some arguments in this identity to $0,\pi/2,\pi/2+\pi \tau/2,$ or $\pi \tau/2$ leads to various addition formulas. A combination of such specializations with \eqref{t1d} and \eqref{t4d} allows one to derive the identities \cite[\S 8.119]{gr}:
\begin{align}
\label{ta4412}\theta_4(0|\tau)^2\theta_4(x+y|\tau)\theta_4(x-y|\tau)&=\theta_4(x|\tau)^2\theta_4(y|\tau)^2-\theta_1(x|\tau)^2\theta_1(y|\tau)^2\\
\label{ta3322}&=\theta_3(x|\tau)^2\theta_3(y|\tau)^2-\theta_2(x|\tau)^2\theta_2(y|\tau)^2,
\end{align}
\begin{equation}
\label{ta1441}\theta_1(x+y|\tau)\theta_2(x-y|\tau)=\theta_1(2x|2\tau)\theta_4(2y|2\tau)
+\theta_4(2x|2\tau)\theta_1(2y|2\tau).
\end{equation}
Moreover, we need the differential identity \cite[\S 8.199]{gr}
  \begin{equation}\label{32d}\frac{\partial}{\partial u}\frac{\theta_3(u|\tau)}{\theta_2(u|\tau)}=\frac{\theta_4(0|\tau)^2\theta_1(u|\tau)\theta_4(u|\tau)}{\theta_2(u|\tau)^2}.\end{equation}

%

By  \eqref{zge} with $\eta=\pi/3$, the parameter $\zeta=\zeta(\tau)$ in \eqref{jsp} is given by
\begin{subequations}\label{zcc}
\begin{equation}\label{zss}\zeta=\frac{\theta_1(2\pi/3|2\tau)^2}{\theta_4(2\pi/3|2\tau)^2}=\frac{\theta_1(\pi/3|\tau)^2\theta_2(\pi/3|\tau)^2}{\theta_3(\pi/3|\tau)^2\theta_4(\pi/3|\tau)^2},\end{equation}
where the second identity follows from \eqref{t1d} and \eqref{t4d}. 
We will also need
\begin{align}
\label{za}1+\zeta&=2\frac{\theta_2(\pi/3|\tau)\theta_3(0|\tau)}{\theta_2(0|\tau)\theta_3(\pi/3|\tau)},\\
\label{zb}1-\zeta&=2\frac{\theta_2(\pi/3|\tau)\theta_4(0|\tau)}{\theta_2(0|\tau)\theta_4(\pi/3|\tau)},\\
\label{zc}3+\zeta&=2\frac{\theta_1(\pi/3|\tau)^2\theta_4(0|\tau)\theta_4(\pi/3|\tau)}{\theta_2(0|\tau)\theta_2(\pi/3|\tau)\theta_3(\pi/3|\tau)^2},\\
\label{zd}3-\zeta&=2\frac{\theta_1(\pi/3|\tau)^2\theta_3(0|\tau)\theta_3(\pi/3|\tau)}{\theta_2(0|\tau)\theta_2(\pi/3|\tau)\theta_4(\pi/3|\tau)^2}.
\end{align}
\end{subequations}
These can, for instance, be obtained from
\cite[Lemma 9.1]{r3}. 
As a consequence,
$$
\gamma=\frac{\zeta+3}{\zeta-1}=-\frac{\theta_1(\pi/3|\tau)^2\theta_4(\pi/3|\tau)^2}{\theta_2(\pi/3|\tau)^2\theta_3(\pi/3|\tau)^2}.
$$
Yet another application of \cite[Lemma 9.1]{r3} then gives
 \begin{equation}\label{gz}\gamma^2=\frac{(1-z)(1+2z)}{1+z}, \end{equation}
 where $z$ is as in \eqref{ztau}.

We now return to the situation when $\eta$ is generic.

\begin{lemma}\label{ssjl}
If $J_x$, $J_y$ and $J_z$ are given by \eqref{aip}, with $\zeta$ and $\Gamma$ parametrized
as in \eqref{zg}, then
\begin{equation}\label{jsg}J_xJ_y+J_xJ_z+J_yJ_z=1-\zeta^2+2\Gamma=
\frac{\theta_1(3\eta|\tau)\theta_3(0|\tau)^4\theta_4(0|\tau)^4}{\theta_1(\eta|\tau)\theta_3(\eta|\tau)^4\theta_4(\eta|\tau)^4}
.\end{equation}
\end{lemma}

Note that this implies Baxter's observation that \eqref{jss} holds for $\eta=\pi/3$.

\begin{proof}
The application of \eqref{tdpe} and \eqref{ta4412} to \eqref{zge} yields
$$1-\zeta^2=\frac{\theta_4(0|2\tau)^3\theta_4(4\eta|2\tau)}{\theta_4(2\eta|2\tau)^4}, \qquad 2\Gamma=\frac{\theta_4(0|2\tau)^3\theta_1(4\eta|2\tau)}{\theta_4(2\eta|2\tau)^3\theta_1(2\eta|2\tau)}. $$
Applying \eqref{ta1441}, we obtain
$$1-\zeta^2+2\Gamma
=\frac{\theta_4(0|2\tau)^3\theta_1(3\eta|\tau)\theta_2(\eta|\tau)}{\theta_1(2\eta|2\tau)\theta_4(2\eta|2\tau)^4}.$$
Using \eqref{t1d} and \eqref{t4d}, this can be written in the desired form.
\end{proof}

\begin{lemma}\label{fdl}
The identity \eqref{zze} holds.
\end{lemma}

\begin{proof}
We differentiate \eqref{jsg} with respect to $\eta$. Setting $\eta=\pi/3$ and applying \eqref{tdps} leads to
\begin{equation}\label{zzege}\left.\zeta\zeta_\eta-\Gamma_\eta\right|_{\eta=\pi/3}=\frac 32\frac{\theta_1'(0|\tau)\theta_3(0|\tau)^4\theta_4(0|\tau)^4}{\theta_1(\pi/3|\tau)\theta_3(\pi/3|\tau)^4\theta_4(\pi/3|\tau)^4}=24\frac{\theta_1'(0|\tau)\theta_2(\pi/3|\tau)^4}{\theta_1(\pi/3|\tau)\theta_2(0|\tau)^4}. \end{equation}
Moreover, we evaluate the ratio $b_u/a$ at $u=\eta=\pi/3$ by differentiating \eqref{bw} with respect to the spectral parameter $u$. Applying \eqref{t1d}, \eqref{t4d}  and \eqref{tdps} yields
\begin{align*}\frac{b_u}{a}\bigg|_{u=\eta=\pi/3}&=\frac{\theta_1'(0|2\tau)\theta_4(2\pi/3|2\tau)}{\theta_4(0|2\tau)\theta_1(2\pi/3|2\tau)}=\frac{\theta_1'(0|\tau)\theta_2(0|\tau)\theta_3(\pi/3|\tau)\theta_4(\pi/3|\tau)}{2\,\theta_3(0|\tau)\theta_4(0|\tau)\theta_1(\pi/3|\tau)\theta_2(\pi/3|\tau)}\\
&=\frac{\theta_1'(0|\tau)\theta_2(0|\tau)^2}{4\,\theta_1(\pi/3|\tau)\theta_2(\pi/3|\tau)^2}.
 \end{align*}
Combining these two expression gives the desired result.
\end{proof}

 \begin{lemma}\label{til}
 In the notation introduced in the main text,
 $$2n\left(\frac{\phi''}{\phi}-\left(\frac{\phi'}{\phi}\right)^2\right)\bigg|_{u=\pi/3}+(2n+1)E=-(2n+1)\frac{\chi(\gamma^2-3)}{(\gamma+1)^2}.$$
 \end{lemma} 

\begin{proof}
Recall that $\phi=\theta_1^{2n+1}$, where $\theta_1=\theta_1(u|\tau)$. Hence,
$$ \frac{\phi''}{\phi}-\left(\frac{\phi'}{\phi}\right)^2=(\log\theta_1^{2n+1})''
=(2n+1)\left(\frac{\theta_1''}{\theta_1}-\left(\frac{\theta_1'}{\theta_1}\right)^2\right).$$
It is  clear from \eqref{ed} that
$$E=2nF+G, $$
where
$$\frac{\theta_1(u|\tau)}{\theta_1(3u|3\tau)}=F_0+F_0F u^2+\mathcal O(u^4), $$
$$\frac{\theta_1(u|\tau)^3}{\theta_4(3u|3\tau)}=G_0u^3+G_0G u^5+\mathcal O(u^7). $$
Thus, it is enough to show that
\begin{equation}\label{tia}F=-\frac{\theta_1''}{\theta_1}+\left(\frac{\theta_1'}{\theta_1}\right)^2
\Bigg|_{u=\frac\pi 3},\end{equation}
\begin{equation}\label{tib}G=-\chi\frac{\gamma^2-3}{(\gamma+1)^2}
=\chi\frac{\zeta^2-6\zeta-3}{2(\zeta+1)^2}. 
\end{equation}

By \eqref{tripl},
$$\frac{\theta_1(u|\tau)}{\theta_1(3u|3\tau)}=\frac{C}{\theta_1\left(\frac\pi3+u|\tau\right) \theta_1\left(\frac\pi3-u|\tau\right)}, $$
where $C$ is independent of $u$. On the other hand, by direct Taylor expansion,
$$ \frac{\theta_1\left(\frac\pi3+u|\tau\right) \theta_1\left(\frac\pi3-u|\tau\right)}{\theta_1\left(\frac\pi3|\tau\right)^2}
=1+\left(\frac{\theta_1''}{\theta_1}-\left(\frac{\theta_1'}{\theta_1}\right)^2\right)\Bigg|_{u=\frac\pi 3} u^2+\mathcal O(u^4).
$$
This proves \eqref{tia}.

  To prove \eqref{tib},  we consider the function
  $$f(u)=\frac{\theta_4\left(\frac\pi 3+u\big|\tau\right)\theta_4\left(\frac\pi 3-u\big|\tau\right)}{\theta_2(u|\tau)\theta_3(u|\tau)}.$$
  On the one hand, it follows from \eqref{tdpe} and  \eqref{tripl} that
  $$f(u)=C \frac{\theta_4(3u|3\tau)\theta_1(u|\tau)}{\theta_1(2u|\tau)},$$
  with $C$ independent of $u$.
  This implies
  \begin{equation}\label{tfd}\frac{f''(0)}{2f(0)}=-G.\end{equation}
  On the other hand, it follows from \eqref{ta3322} that
  $$f(u)=\frac{1}{\theta_4(0|\tau)^2 }\left( \theta_3\left(\frac\pi3\Big|\tau\right)^2\frac{\theta_3(u|\tau)}{\theta_2(u|\tau)}- \theta_2\left(\frac\pi3\Big|\tau\right)^2\frac{\theta_2(u|\tau)}{\theta_3(u|\tau)}\right).$$
  Differentiating this identity using \eqref{32d}
 gives
  $$f'(u)=\left(\frac{\theta_3(\frac{\pi}{3}\big|\tau)^2}{\theta_2(u|\tau)^2}+\frac{\theta_2(\frac{\pi}{3}\big|\tau)^2}{\theta_3(u|\tau)^2}\right)\theta_1(u|\tau)\theta_4(u|\tau)$$
  and hence
  $$f''(0)= \left(\frac{\theta_3(\frac{\pi}{3}\big|\tau)^2}{\theta_2(0|\tau)^2}+\frac{\theta_2(\frac{\pi}{3}\big|\tau)^2}{\theta_3(0|\tau)^2}\right)\theta_1'(0|\tau)\theta_4(0|\tau).$$
  Using also \eqref{jtv}, it follows that
  \begin{equation}\label{fsdl}\frac{f''(0)}{f(0)}=\frac{\theta_1'(0|\tau)^2}{2\theta_4(\pi/3|\tau)^2} \left(\frac{\theta_3(\pi/3|\tau)^2}{\theta_2(0|\tau)^2}+\frac{\theta_2(\pi/3|\tau)^2}{\theta_3(0|\tau)^2}\right).\end{equation}
  We now use \eqref{zcc} to write
 \begin{align*} \frac{\theta_3(\pi/3|\tau)^2}{\theta_2(0|\tau)^2}&=\frac{\theta_2(\pi/3|\tau)^2\theta_4(\pi/3|\tau)^2}{\theta_1(\pi/3|\tau)^2\theta_2(0|\tau)^2}\frac{3-\zeta}{1+\zeta},\\
  \frac{\theta_2(\pi/3|\tau)^2}{\theta_3(0|\tau)^2}&=\frac{\theta_2(\pi/3|\tau)^2\theta_4(\pi/3|\tau)^2}{\theta_1(\pi/3|\tau)^2\theta_2(0|\tau)^2}\frac{4\zeta}{(1+\zeta)^2}.\end{align*}
  Inserting these expressions in \eqref{fsdl} and comparing with \eqref{tfd} gives
  \eqref{tib}.
 \end{proof} 
 
 \begin{lemma}\label{chkl}
 The identity \eqref{hki} holds.
 \end{lemma}
 
 \begin{proof}
 Since $\gamma+1=2(\zeta+1)/(\zeta-1)$, it follows from \eqref{tde}, \eqref{tdf}
 and \eqref{zcc}
  that
 \begin{align*}\frac{(\gamma+1)^2}\chi&=4\left(\frac{\theta_2(0|\tau)\theta_3(0|\tau)\theta_1(\pi/3|\tau)\theta_4(\pi/3|\tau)}{\theta_1'(0|\tau)\theta_4(0|\tau)\theta_2(\pi/3|\tau)\theta_3(\pi/3|\tau)}\right)^2\\
& =4\left(\frac{\theta_2(0|\tau/2)\theta_1(\pi/3|\tau/2)}{\theta_1'(0|\tau/2)\theta_2(\pi/3|\tau/2)}\right)^2. 
\end{align*}
 Using  \eqref{tripl}, it is easy to verify that
 $$\frac{\theta_1'(0|3\tau/2)}{\theta_2(0|3\tau/2)}=\frac 13\frac{\theta_1'(0|\tau/2)\theta_1(\pi/3|\tau/2)^2}{\theta_2(0|\tau/2)\theta_2(\pi/3|\tau/2)^2}, $$
 $$\frac{h(u+\pi+\frac{\pi\tau}2)}{h(u+\frac{\pi\tau}2)}
 =\ti\frac{\theta_1(\frac u2|\frac\tau2)}{\theta_2(\frac u2|\frac\tau2)}
 \left(\frac{\theta_2(\frac\pi3+\frac u2|\frac\tau 2)\theta_2(\frac\pi3-\frac u2|\frac\tau 2)}{\theta_1(\frac\pi3+\frac u2|\frac\tau 2)\theta_1(\frac\pi3-\frac u2|\frac\tau 2)}\right)^n,
  $$
  $$\frac{k(u)}{k(u+\pi)}=\left(\frac{\theta_2(\frac\pi3+\frac u2|\frac\tau 2)\theta_2(\frac\pi3-\frac u2|\frac\tau 2)}{\theta_1(\frac\pi3+\frac u2|\frac\tau 2)\theta_1(\frac\pi3-\frac u2|\frac\tau 2)}\right)^2\left(\frac{\theta_4(\frac\pi3+\frac u2|\frac\tau 2)\theta_4(\frac\pi3-\frac u2|\frac\tau 2)}{\theta_3(\frac\pi3+\frac u2|\frac\tau 2)\theta_3(\frac\pi3-\frac u2|\frac\tau 2)}\right)^{n-2}$$
  and hence
 $$\frac{h'(\pi+\frac{\pi\tau}2)}{h(\frac{\pi\tau}2)}=\frac\ti 2\frac{\theta_1'(0|\frac\tau2)}{\theta_2(0|\frac\tau2)}
 \left(\frac{\theta_2(\frac\pi3|\frac\tau 2)}{\theta_1(\frac\pi3|\frac\tau 2)}\right)^{2n}, $$
  $$\frac{k(0)}{k(\pi)}=\left(\frac{\theta_2(\frac\pi3|\frac\tau 2))}{\theta_1(\frac\pi3|\frac\tau 2)}\right)^4\left(\frac{\theta_4(\frac\pi3|\frac\tau 2)}{\theta_3(\frac\pi3|\frac\tau 2)}\right)^{2n-4}. $$
  Finally, it follows from  \cite[Lemma 9.1]{r3} that
  $$2z+1=-\left(\frac{\theta_1(\frac\pi 3|\frac\tau2)\theta_3(\frac\pi 3|\frac\tau2)}{\theta_2(\frac\pi 3|\frac\tau2)\theta_4(\frac\pi 3|\frac\tau2)}\right)^2, $$
  $$\frac{z+1}{z-1}=\left(\frac{\theta_2(\frac\pi 3|\frac\tau2)\theta_3(\frac\pi 3|\frac\tau2)}{\theta_1(\frac\pi 3|\frac\tau2)\theta_4(\frac\pi 3|\frac\tau2)}\right)^2. $$
  Combining all these identities gives the desired result.
  \end{proof}

  \section{Special polynomials}
  \label{sp}
  
In this appendix, we collect some identities related to the polynomials
$s_n$ and $\bar s_n$. In \cite[\S 5.2]{r4}, the second author identified them as
a special case of the more general rational functions $t^{(k_0,k_1,k_2,k_3)}$.  
Replacing $\zeta$ by $-z-1$ in the relevant identities gives
  \begin{subequations}\label{sst}
 \begin{align}s_n\left(\gamma^{-2}\right)&=\left(\frac{(-1)(z-1)}{2z^2(z+1)^2(2z+1)}\right)^{\frac{n(n-1)}2} t^{(n,n,0,0)}(-z-1), \\
 \bar s_{n}\left(\gamma^{-2}\right)&=
\frac{(-1)^{\frac{(n+1)(n+2)}2}(z-1)^{\frac{(n-1)(n+2)}2}}{(2z^2(2z+1))^{\frac{n(n-1)}2}(z+1)^{n^2-1}}\,t^{(n,n,1,-1)}(-z-1),\end{align}
\end{subequations}
where $\gamma$ and $z$ are related as in \eqref{gz}.

We will need further relations that can be obtained from \eqref{sst} using symmetries of the $t$-polynomials. First of all,
applying 
\cite[Prop.\ 2.3]{r4}  gives
\begin{subequations}\label{smt}
\begin{align}
\nonumber s_{-n-1}\left(\gamma^{-2}\right)&=
\frac{(-1)^{\frac{(n-1)(n-2)}2}3^n(z-1)^{\frac{n^2-5n+2}2}}{2^{\frac{n^2-n+2}2}z^{n(n-5)}(z+1)^{(n-1)(n-2)}(z+2)^2(2z+1)^{\frac{n^2-n+2}2}}\\
&\label{sstc}\quad\times
 t^{(n,n,-1,-1)}(-z-1),\\
 \nonumber \bar s_{-n-1}\left(\gamma^{-2}\right)&=
\frac{(-1)^{\frac{n(n+1)}2}3^n(z-1)^{\frac{n^2-7n+2}2}}{2^{\frac{n^2-n+2}2}z^{n(n-5)}(z+1)^{n^2-4n+2}(z+2)^2(2z+1)^{\frac{n^2-n+6}2}}\\
&\label{sstd}\quad\times
 t^{(n,n,-2,0)}(-z-1).
\end{align}
\end{subequations}
It folllows from \cite[Prop.\ 2.2]{r4} that
$$t^{(k_0,k_1,k_2,k_3)}(-z-1)=\left(\frac{(-1)^{k_1+k_3}(z+1)^{k_1+k_2-n}(z+2)^{k_2+k_3}}{z^{k_2+k_3-n}(z-1)^{k_1+k_2}}\right)^{n-1} t^{(k_0,k_3,k_2,k_1)}(z),$$
where $n=(k_0+k_1+k_2+k_3)/2$.
Applying this transformation to \eqref{sst} and \eqref{sstc} gives
\begin{subequations}\label{stfc}
\begin{align}s_n\left(\gamma^{-2}\right)&=\left(\frac{(-1)}{2(z+1)^2(z-1)(2z+1)}\right)^{\frac{n(n-1)}2} t^{(n,0,0,n)}(z), \\
 \bar s_n\left(\gamma^{-2}\right)&=\frac{(-1)^{\frac{n(n+1)}2}}{\big(2(z+1)^2(z-1)(2z+1)\big)^{\frac{n(n-1)}2}} \,t^{(n,-1,1,n)}(z), \\
\nonumber s_{-n-1}\left(\gamma^{-2}\right)&=
 \frac{(-1)^{\frac{(n-1)(n-2)}2}3^nz^{4n-2}}{(2(z-1)(2z+1))^{\frac{n^2-n+2}2}(z+1)^{(n-1)(n-2)}(z+2)^{2n-2}}\\
 &\quad\times t^{(n,-1,-1,n)}(z). \end{align}
Finally,
 combining \cite[Prop.\ 2.2]{r4} and \cite[Prop.\ 2.4]{r4}
gives
\begin{multline*}t^{(k_0,k_1,k_2,k_3)}(-z-1)=(-1)^{(k_1+k_2+n)(k_1+k_3+n)}
\frac{Y_{n-k_0}Y_{n-k_1}Y_{n-k_2}Y_{n-k_3}}{Y_{k_0}Y_{k_1}Y_{k_2}Y_{k_3}}\\
\times\left(\frac{z^{k_0+k_1-n}(z+1)^{k_1+k_2-n}(z+2)^{k_2+k_3}}{(z-1)^{k_1+k_2}(2z+1)^{k_1+k_3-n}}\right)^{n-1}
 t^{(n-k_2,n-k_1,n-k_0,n-k_3)}(z),\end{multline*}
where again $n=(k_0+k_1+k_2+k_3)/2$ and $Y_k$ are elementary coefficients satisfying
$$Y_{k+1}Y_{k-1}=2(2k+1)Y_k^2,\qquad Y_0=Y_1=1. $$
Applying this to \eqref{sstd} gives
  \begin{align}\nonumber\bar s_{-n-1}\left(\gamma^{-2}\right)&=
 \frac{(-1)^{\frac{n(n+1)}2}3^n(2n+1)z^{4n-2}}{2^{\frac{n^2-n+2}2}(z+1)^{n(n-3)}(z-1)^{\frac{n^2-n+6}2}(z+2)^{2n-2}(2z+1)^{\frac{n^2+n+2}2}}\\
&\quad\times  t^{(n+1,-1,-1,n-1)}(z).\end{align}
 \end{subequations}

\section{Infinite-lattice limit}
\label{ills}

Correlation functions for the infinite-length XYZ spin chain have been studied by various methods, see e.g.\ \cite{b4,bj,lp1,lp2,q,sh}. In this appendix, we will sketch how such correlation functions can be computed in the special case of interest to us. We have not made a detailed comparison with more general results available in the literature, but the remark at the end of \cite[\S 4.1]{bj} indicates that our expression should agree with \cite[Eq.\ (4.2)]{bj}, after a suitable identification of parameters.

In this appendix, we write the lowest eigenvalue of the Hamiltonian for a chain of length $L$  as $L\varepsilon_L$ and use $\varepsilon$ for the infinite-lattice limit $\varepsilon=\lim_{L\rightarrow\infty}\varepsilon_L$ (assuming convergence). We have shown that when $L=2n+1$, our correlations can be expressed in terms of the quantity \eqref{feg}, $$f_n=\frac{\varepsilon_L(2(\varepsilon_L)_\eta+\Gamma_\eta)}{\zeta\zeta_\eta-\Gamma_\eta}\Bigg|_{\eta=\pi/3}.$$ If we let $n\rightarrow\infty$, we find that identical expressions \eqref{caj} hold for the correlations, but with $f_n$ replaced by
\begin{equation}\label{finf}
f_\infty=\lim_{n\rightarrow \infty} f_n=\frac{\varepsilon(2\varepsilon_\eta+\Gamma_\eta)}{\zeta\zeta_\eta-\Gamma_\eta}\Bigg|_{\eta=\pi/3}.
\end{equation}
We will compute $f_\infty$ directly using Baxter's explicit expression for $\varepsilon$. We note that Baxter obtained this expression, and several of its properties, by considering the infinite-lattice limit along chains of even length. In contrast, we focus on odd lengths. Hence, we need to assume that the limits along chains of even and odd lengths lead to the same result. Although this assumption seems plausible, we have no proof that it holds.

We start from  \cite[(10.14.30)]{b5}, written in Baxter's notation as
\begin{multline}\label{feb}\varepsilon=\frac{-J_x+J_y+J_z}2\\
-\frac{\pi(J_x^2-J_y^2)^{1/2}}{2I\sqrt{k}}\sum_{m=1}^\infty\frac{x^{-m}(x^{3m}-q^{m/2})(1-q^{m/2}x^{-m})(1-x^{2m})}{(1-q^m)(1+x^{2m})}. \end{multline}
Here, $I,k,q,x$ are functions of the spin chain's anisotropy parameters. The expression for $\varepsilon$ holds if these parameters are in the so-called principal regime, defined through the inequalities
$$|J_y|<J_x<-J_z. $$
The values in other regimes follow from the fact that $\varepsilon$ is invariant if the parameters $(J_x,J_y,J_z)$ are permuted or if any two of them are multiplied by $-1$.

We parametrize the chain as in \eqref{zge}, \eqref{aip}, where we now take $\eta\in\mathbb R$ and $\tau\in \ti\mathbb R_{>0}$. One can show that $0<\zeta<1$ and hence 
$|J_z|<J_y<J_x$. Thus, we should apply \eqref{feb} with the replacements $(J_x,J_y,J_z)\mapsto(J_y,-J_z,-J_x)$. The parameters in \eqref{feb} are related to ours by
$$q^{1/2}=e^{-\pi\ti/\tau},\qquad x=e^{\ti(2\eta-\pi)/\tau} $$
and one can compute
$$\frac{\pi}{2I\sqrt{k}}=\frac{2\ti}{\tau\theta_4(0|\tau)^2},
\qquad (J_y^2-J_z^2)^{1/2}=\frac{2\theta_4(0|\tau)^2\theta_1(2\eta|2\tau)}{\theta_2(0|\tau)^2\theta_4(2\eta|2\tau)}. $$
The replacements, together with these expressions, lead us to the identity
\begin{multline}\label{bfe}
  \varepsilon=-\frac{J_x+J_y+J_z}2\\
  -\frac{4\ti\theta_1(2\eta|2\tau)}{\tau\theta_2(0|\tau)^2\theta_4(2\eta|2\tau)}\sum_{m=1}^\infty\frac{x^{-m}(x^{3m}-q^{m/2})(1-q^{m/2}x^{-m})(1-x^{2m})}{(1-q^m)(1+x^{2m})}.  
\end{multline}
The supersymmetric case $\eta=\pi/3$ corresponds to
$x^3=q^{1/2}$. In this case, all terms in the infinite series vanish and we recover \eqref{gse}, that is,
\begin{equation}\label{rsr}\varepsilon\big|_{\eta=\pi/3}=-\frac{J_x+J_y+J_z}{2}=-\frac{\zeta^2+3}{4}.
\end{equation}

Let us now consider the quantity $2\varepsilon_\eta+\Gamma_\eta\big|_{\eta=\pi/3}$. Since $\Gamma=J_x+J_y+J_z-2$, the contribution from the first term on the right-hand side of \eqref{bfe} vanishes. In the remaining terms, the derivative with respect to $\eta$ must hit the factor $x^{3m}-q^{m/2}$ to yield a non-zero contribution when $\eta=\pi/3$. It follows that
\begin{equation}\label{egd} 2\varepsilon_\eta+\Gamma_\eta\Big|_{\eta=\pi/3}
=\frac{48\,\theta_1(2\pi/3|2\tau)}{\tau^2\theta_2(0|\tau)^2\theta_4(2\pi/3|2\tau)}
\,S,\end{equation}
where
$$S=
\sum_{m=1}^\infty\frac{mq^{\frac m3}(1-q^{\frac m3})^2}{(1-q^m)(1+q^{\frac m3})}.
$$

To compute $S$ we will use the identity \cite[\S 21, Ex.\ 11]{ww}
$$\frac{\partial}{\partial u}\log\theta_4(u|\tau)
=4\sum_{m=1}^{\infty}\frac{e^{\ti\pi\tau  m}\sin(2mu)}{1-e^{2\ti\pi\tau m}}, \qquad |\operatorname{Im}(u)|<\frac{\pi}{2}\,\operatorname{Im}(\tau),$$
which implies
$$\frac{\partial^2}{\partial u^2}\log\theta_4(u|\tau)
=8\sum_{m=1}^{\infty}\frac{me^{\ti\pi\tau  m}\cos(2mu)}{1-e^{2\ti\pi\tau m}}. $$
Rewriting \eqref{egd} with the help of the identity
\begin{align*}\frac{q^{\frac m3}(1-q^{\frac m3})^2}{(1-q^m)(1+q^{\frac m3})}
&=q^m\frac{4-3(q^{\frac m3}+q^{-\frac m3})+(q^{\frac{2m}3}+q^{-\frac{2m}3})}{1-q^{2m}}\\
&=e^{-2\pi\ti m/\tau}\frac{4-6\cos\left(\frac{2\pi m}{3\tau}\right)+2\cos\left(\frac{4\pi m}{3\tau}\right)}{1-e^{-4\pi\ti m/\tau}}, \end{align*}
it follows that 
$$ S=\frac 14
\left.\frac{\partial^2}{\partial u^2}\log\frac{\theta_4\big(u|-\frac 2\tau\big)^2\theta_4\big(u+\frac{2\pi}{3\tau}|-\frac 2\tau\big)}{\theta_4\big(u+\frac \pi{3\tau}|-\frac 2\tau\big)^3}\right|_{u=0}.
$$
Applying \eqref{tmt} gives
$$\label{slg} S=\frac {\tau^2}{16}
\left.\frac{\partial^2}{\partial u^2}\log g(u)\right|_{u=0},
$$
where
$$g(u)=\frac{\theta_2\big(u|\frac\tau2\big)^2\theta_2\big(u+\frac{\pi}{3}|\frac \tau2\big)}{\theta_2\big(u+\frac \pi6|\frac \tau2\big)^3}. $$
Inserting this expression in \eqref{egd}, we obtain
\begin{equation}
  \label{edlg}
  2\varepsilon_\eta+\Gamma_\eta\Big|_{\eta=\pi/3}=\frac{3\theta_1(2\pi/3|2\tau)}{\theta_2(0|\tau)^2\theta_4(2\pi/3|2\tau)}\left.\frac{\partial^2}{\partial u^2}\log g(u)\right|_{u=0}.
\end{equation}

To proceed, we use some theta function identities. On the one hand, we have
$$\left.\frac{\partial^2}{\partial u^2}\log g(u)\right|_{u=0}=\frac{g''(0)g(0)-g'(0)^2}{g(0)^2}=\left.\frac 1{2g(0)^2}\frac{\partial^2}{\partial u^2} g(u)g(-u)\right|_{u=0}.$$ On the other hand, using \eqref{wsi}, one can write $$g(u)g(-u)=\frac{\theta_1(\pi/3|\tau/2)^2}{\theta_2(\pi/3|\tau/2)^2}\,h(u)^2-\frac{\theta_2(0|\tau/2)}{\theta_2(\pi/3|\tau/2)}\,h(u)^3, $$
where
$$h(u)=\frac{\theta_2\big(u|\frac\tau 2\big)^2}{\theta_1\big(\frac \pi3+u|\frac \tau2\big)\theta_1\big(\frac \pi3-u|\frac \tau2\big)}. $$
  Since $h$ is even, it follows that
  \begin{equation}\label{sldg}\left.\frac{\partial^2}{\partial u^2}\log g(u)\right|_{u=0}
  =\frac{h(0)h''(0)}{2g(0)^2}\left(2\frac{\theta_1(\pi/3|\tau/2)^2}{\theta_2(\pi/3|\tau/2)^2}-3\frac{\theta_2(0|\tau/2)}{\theta_2(\pi/3|\tau/2)}\,h(0)\right).\end{equation}
  Again by \eqref{wsi},
  $$h(u)-h(0)=\frac{\theta_2\big(\frac\pi3|\frac\tau2\big)^2\theta_1\big(u|\frac\tau2\big)^2}{\theta_1\big(\frac\pi3|\frac\tau2\big)^2\theta_1\big(\frac \pi3+u|\frac \tau2\big)\theta_1\big(\frac \pi3-u|\frac \tau2\big)}, $$
  which implies 
  $$h''(0)=\frac{2\theta_2\big(\frac\pi3|\frac\tau2\big)^2\theta_1'(0|\frac\tau2)^2}{\theta_1\big(\frac\pi3|\frac\tau2\big)^4}.$$
  Inserting also
  $$g(0)=\frac{\theta_2\big(0|\frac\tau2\big)^2\theta_2\big(\frac{\pi}{3}|\frac \tau2\big)}{\theta_1\big(\frac \pi 3|\frac \tau2\big)^3},
  \qquad h(0)=\frac{\theta_2(0|\frac\tau 2)^2}{\theta_1\big(\frac \pi3|\frac \tau2\big)^2} $$
  in \eqref{sldg} gives
$$
  \frac{\partial^2}{\partial u^2}\log g(u)\Bigg|_{u=0}=
  \frac{\theta_1'(0|\frac\tau2)^2}{\theta_2(0|\frac\tau2)^2}\left(
  2\frac{\theta_1\big(\frac\pi3|\frac\tau2\big)^2}{\theta_2\big(\frac\pi3|\frac\tau2\big)^2}-3\frac{\theta_2(0|\frac\tau2)^3}{\theta_1\big(\frac\pi3|\frac\tau2\big)^2\theta_2\big(\frac\pi3|\frac\tau2\big)}
  \right).$$
  Using first \eqref{tde}--\eqref{tdf} and then \eqref{zcc}, one checks that
  \begin{align*}
  2\frac{\theta_1\big(\frac\pi3|\frac\tau2\big)^2}{\theta_2\big(\frac\pi3|\frac\tau2\big)^2}-3\frac{\theta_2(0|\frac\tau2)^3}{\theta_1\big(\frac\pi3|\frac\tau2\big)^2\theta_2\big(\frac\pi3|\frac\tau2\big)}
  & =2\frac{3+\zeta}{1-\zeta}-24\frac{1+\zeta}{(3+\zeta)(1-\zeta)}\\
  &=\frac{2(\zeta^2-6\zeta-3)}{(1-\zeta)(3+\zeta)}.
  \end{align*}
We insert this expression in \eqref{edlg} and arrive at
\begin{equation}\label{egdf}  2\varepsilon_\eta+\Gamma_\eta\Big|_{\eta=\pi/3}
=6\frac{\theta_1(2\pi/3|2\tau)}{\theta_2(0|\tau)^2\theta_4(2\pi/3|2\tau)}
\frac{\theta_1'(0|\frac\tau2)^2}{\theta_2(0|\frac\tau2)^2}
\frac{(\zeta^2-6\zeta-3)}{(1-\zeta)(3+\zeta)}.
\end{equation}

  Returning to \eqref{finf}, we insert \eqref{zzege}, \eqref{rsr} and \eqref{egdf} to obtain
  $$f_\infty=-\frac 18\frac{\theta_1(2\pi/3|2\tau)\theta_1(\pi/3|\tau)\theta_2(0|\tau)^2\theta_1'(0|\tau/2)^2}{\theta_4(2\pi/3|2\tau)\theta_1'(0|\tau)\theta_2(\pi/3|\tau)^4\theta_2(0|\tau/2)^2}\frac{(\zeta^2+3)(\zeta^2-6\zeta-3)}{(1-\zeta)(3+\zeta)}. $$
  Using \eqref{t1d}--\eqref{tdf}, we rewrite this expression in terms of theta functions with nome $\tau$. Simplifying the resulting expression with the help of  \eqref{jtv} and \eqref{tdps} gives
  $$ f_\infty=-\frac14\frac{\theta_1(\pi/3|\tau)^2\theta_4(0|\tau)^2}{\theta_2(\pi/3|\tau)^2\theta_3(0|\tau)^2}\frac{(\zeta^2+3)(\zeta^2-6\zeta-3)}{(1-\zeta)(3+\zeta)}. $$
  We can then apply \eqref{zcc} and finally obtain
  \begin{equation}\label{finfres}f_\infty=-\frac{(\zeta^2+3)(\zeta^2-6\zeta-3)}{8(\zeta+1)^2}=\frac{(\gamma^2+3)(\gamma^2-3)}{(\gamma^2-1)^2},
  \end{equation}
where $\gamma=(\zeta+3)/(\zeta-1)$.

The identity \eqref{finfres} holds for $0<\zeta<1$. The expression remains valid for $\zeta=0$ by Stroganov's results \cite{st,st2}, and for $\zeta=1$, too, thanks to the triviality of the correlation functions at this point. The value in other regimes follows since $f_\infty$ is invariant under permutations of the anisotropy parameters, that is, it satisfies the same symmetries \eqref{fns} as $f_n$. This gives the final result \eqref{finfe}.
We repeat that, in the parametrization \eqref{jsp}, the nearest-neighbour correlation functions for the infinite chain are given by \eqref{cf} with $f_n$ replaced by $f_\infty$. For instance,
$$C^z=1-\frac{4f_\infty}{\zeta^2+3}=\begin{cases}\displaystyle\frac{\zeta^4-6\zeta^2+13}{(\zeta^2-1)^2}, & |\zeta|\geq 3,\\[3mm]
\displaystyle\frac{(\zeta+1)(3\zeta-1)}{2(\zeta-1)^2}, & -3\leq\zeta\leq 0,\\[3mm]
\displaystyle\frac{(\zeta-1)(3\zeta+1)}{2(\zeta+1)^2},& 0\leq \zeta\leq 3.\\
\end{cases} $$
One can check that the correlation functions are twice differentiable in $\zeta$, but their third derivatives jump at the XXZ points $\zeta=0$ and $\zeta=\pm 3$.


\begin{thebibliography}{999}
\bibitem[B1]{b1} R.\ J.\ Baxter, \emph{Partition function of the eight-vertex lattice model}, Ann.\ Phys.\ 70 (1972), 193--228. 
\bibitem[B2]{b2} R.\ J.\ Baxter, \emph{One-dimensional anisotropic Heisenberg chain}, Ann.\ Phys.\ 70 (1972), 323--337. 
\bibitem[B3]{b3} R.\ J.\ Baxter, \emph{Eight-vertex model in lattice statistics and one-dimensional anisotropic Heisenberg chain. I. Some fundamental eigenvectors}.
Ann.\ Phys.\ 76 (1973), 1--24.
\bibitem[B4]{b4} R.\ J.\ Baxter, \emph{Solvable eight-vertex model on an arbitrary planar lattice}, Philos.\ Trans.\ Roy.\ Soc.\ London Ser.\ A 289 (1978), 315--346.
 \bibitem[B5]{b5} R.\ J.\ Baxter, Exactly Solved Models in Statistical Mechanics, Academic Press,  1982.
\bibitem[BK]{bk} R.\ J.\ Baxter and S.\ B.\ Kelland, \emph{Spontaneous polarization of the eight-vertex model}, J.\ Phys.\ C: Solid State Phys.\ 7 (1974), L403--406.
 \bibitem[BM1]{bm}  V.\ V.\ Bazhanov and V.\ V.\ Mangazeev, \emph{
Eight-vertex model and non-stationary Lam\'e equation},
J.\ Phys.\ A 38 (2005),  L145--L153. 
\bibitem[BM2]{bm1} V.\ V.\ Bazhanov and V.\ V.\ Mangazeev, \emph{The eight-vertex model and Painlev\'e VI}, J.\ Phys.\ A 39 (2006), 12235--12243. 
\bibitem[BM3]{bm2} V.\ V.\ Bazhanov and V.\ V.\ Mangazeev, \emph{Analytic theory of the eight-vertex model}, Nuclear Phys.\ B 775 (2007), 225--282. 
\bibitem[BM4]{bm4} V.\ V.\ Bazhanov and V.\ V.\ Mangazeev, \emph{
The eight-vertex model and Painlev\'e VI equation II: eigenvector results},
J.\ Phys.\ A 43 (2010), 085206.
\bibitem[BJ]{bj} H.\ Boos, M.\ Jimbo,  Y.\ Miwa, F.\ Smirnov and Y.\ Takeyama, \emph{ Traces on the Sklyanin algebra and correlation functions of the eight-vertex model}, J. Phys. A 38 (2005),  7629--7659.
\bibitem[BH]{bh} S.\ Brasseur and C.\ Hagendorf, \emph{Sum rules for the supersymmetric eight-vertex model}, J.\ Stat.\ Mech.\ Theory Exp.\ (2021), 023102.
\bibitem[C]{ca} L.\ Cantini, \emph{Finite size emptiness formation probability of the XXZ spin chain at
$\Delta=-1/2$},  J.\ Phys.\ A 45 (2012), 135207. 
\bibitem[CS]{cs} L.\ Cantini and A.\ Sportiello, \emph{Proof of the Razumov--Stroganov conjecture}, J.\ Combin.\ Theory Ser.\ A 118 (2011), 1549--1574. 
\bibitem[CD]{cd} 
R.\ Conte and I.\ Dornic,
\emph{The master Painlev\'e VI heat equation},
C.\ R.\ Math.\ Acad.\ Sci.\ Paris 352 (2014),  803--806. 
\bibitem[D]{d} G.\ Darboux, \emph{Sur une \'equation lin\'eaire}, C.\ R.\ Acad.\ Sci.\ Paris 94 (1882), 1645--1648.
\bibitem[F]{f} K.\ Fabricius, \emph{A new $Q$-matrix in the eight-vertex model}, J.\ Phys.\ A 40 (2007), 4075--4086. 
\bibitem[FH]{fh} P.\ Fendley and C.\ Hagendorf, \emph{The eight-vertex model and lattice supersymmetry}, J.\ Stat.\ Phys.\ 146 (2012), 1122--1155. 
\bibitem[FS]{fs} P.\ Fendley and H.\ Saleur, \emph{$N=2$ supersymmetry, Painlev\'e III and exact scaling functions in 2D polymers}, Nuclear Phys.\ B 388 (1992),  609--626. 
\bibitem[FWZ]{fwz} P.\ E.\ Finch, R.\ Weston and P.\ Zinn-Justin,
\emph{
Theta function solutions of the quantum Knizhnik--Zamolodchikov--Bernard equation for a face model},
J.\ Phys.\ A 49 (2016),  064001. 
\bibitem[GR]{gr} I.\ S.\ Gradstein and I.\ M.\ Ryzhik, Table of Integrals, Series and Products, 7\textsuperscript{th} edition, Elsevier, 2007.
\bibitem[HL]{hl} C.\  Hagendorf and J.\ Li\'enardy, \emph{On the transfer matrix of the supersymmetric eight-vertex model. I. Periodic boundary conditions}, J.\ Stat.\ Mech.\ Theory Exp.\ 2018, 033106.
\bibitem[H1]{h1} L.\ Hietala, \emph{A combinatorial description of certain polynomials related to the XYZ spin chain}, SIGMA 16 (2020), 101. 
\bibitem[H2]{h2} L.\ Hietala, \emph{ A combinatorial description of certain polynomials related to the XYZ spin chain. II. The polynomials $p_n$}, SIGMA  18 (2022), 036.
\bibitem[JMN]{jmn} M.\ Jimbo, T.\ Miwa and A.\ Nakayashiki, \emph{Difference equations for the correlation functions of the eight-vertex model}, J.\ Phys.\ A: Math.\ Gen.\ 26 (1993) 2199.
\bibitem[LP1]{lp1} 
M.\ Lashkevich and Y.\ Pugai, \emph{Free field construction for correlation functions of the eight-vertex model}, Nuclear Phys.\ B 516 (1998), 623--651.
 \bibitem[LP2]{lp2} M.\ Lashkevich and Y.\ Pugai, 
 \emph{Nearest-neighbor two-point correlation function of the Z-invariant eight-vertex model}, JETP Lett.\ 68 (1998), 257--262.
\bibitem[L]{l} A.\ Luther, \emph{Eigenvalue spectrum of interacting massive fermions in one dimension}, Phys. Rev. B 14 (1976) 2153–2159.
\bibitem[M]{m} J.\ Malmquist, \emph{Sur les \'equations diff\'erentielles du second ordre, dont l'int\'egrale g\'en\'erale a ses points critiques fixes}, Ark.\ Mat.\ Astron.\ Fys.\ 17 (1923), 1--89.
\bibitem[Ma]{ma} Yu.\ I.\ Manin, \emph{Sixth Painlev\'e equation, universal elliptic curve, and mirror of $P^2$}, in Geometry of Differential Equations, 131--151, Amer.\ Math.\ Soc.,  1998.
\bibitem[MB]{mb} V.\ V.\ Mangazeev and V.\ V.\ Bazhanov,
\emph{The eight-vertex model and Painlev\'e VI equation II: eigenvector results},
J.\ Phys.\ A 43 (2010), 085206. 
\bibitem[NY]{ny} M.\ 
 Noumi and Y.\ Yamada, \emph{A new Lax pair for the sixth Painlev\'e equation associated with $\widehat{\mathfrak{so}}(8)$},
 in  Microlocal Analysis and Complex Fourier Analysis, World Sci.\ Publ.,  2002, 238--252.  
 \bibitem[N]{n} D.\ P.\ Novikov,
\emph{The Schlesinger system with $2\times 2$ matrices and the Belavin--Polyakov--Zamolodchikov equation},
Theoret.\ Math.\ Phys. 161 (2009), 1485--1496. 
\bibitem[O]{o} K.\ Okamoto, \emph{Studies on the Painlev\'e equations.\ I.\ Sixth Painlev\'e equation $P_{{\rm VI}}$}, Ann.\ Mat.\ Pura Appl.\ (4) 146 (1987), 337--381. 
\bibitem[Q]{q}
Y.-H.\  Quano, \emph{Bootstrap equations and correlation functions for the Heisenberg XYZ antiferromagnet}, J.\ Phys.\ A 35 (2002), 9549--9572.
\bibitem[RS1]{rs1} A.\ V.\ Razumov and Yu.\ G.\  Stroganov, 
\emph{Combinatorial nature of the ground-state vector of the O(1) loop model},
Theoret.\ Math.\ Phys.\ 138 (2004),  333--337.
\bibitem[RS2]{rs}
A.\ V.\ Razumov and Yu.\ G.\  Stroganov, 
\emph{A possible combinatorial point for the XYZ spin chain},
Theor.\ Math.\ Phys.\ 164 (2010), 977--991.	
\bibitem[RSZ]{rsz}
A.\ V.\ Razumov, P.\ Zinn-Justin and Yu.\ G.\  Stroganov, 
\emph{Polynomial solutions of $q$KZ equation and ground state of XXZ spin chain at $\Delta=-1/2$},
J.\ Phys.\ A 40 (2007),  11827--11847. 
\bibitem[Ro]{r} S.-S.\ Roan,
\emph{The $Q$-operator and functional relations of the eight-vertex model at root-of-unity $\eta=2mK/N$ for odd $N$}, J.\ Phys.\ A 40 (2007), 11019--11044. 
\bibitem[R1]{r3} H.\ Rosengren, \emph{The three-colour model with domain wall boundary conditions}, Adv.\ Appl.\ Math.\ 46 (2011), 481--535.
\bibitem[R2]{rsa} H.\ Rosengren, \emph{Special polynomials related to the supersymmetric eight-vertex model.\ I.\ Behaviour at cusps}, arXiv:1305.0666.
\bibitem[R3]{rsc} H.\ Rosengren, \emph{Special polynomials related to the supersymmetric eight-vertex model.\ III.\ Painlev\'e VI equation},
arXiv:1405.5318.
\bibitem[R4]{r4} H. Rosengren, \emph{Special polynomials related to the supersymmetric eight-vertex model: A summary}, Comm.\ Math.\ Phys.\ 340 (2015), 1143--1170.
\bibitem[S]{sh} J.\ Shiraishi, \emph{Free field constructions for the elliptic algebra $A_{q,p}(\widehat{\text{sl}}_2)$ and Baxter's eight-vertex model}, Internat.\ J.\ Modern Phys.\ A 19 (2004), 363--380. 
\bibitem[St1]{st} Yu.\ G.\ Stroganov, \emph{The importance of being odd}, J.\ Phys.\ A 34 (2001),  L179--L185.
\bibitem[St2]{st2} Yu.\ G.\ Stroganov, \emph{XXZ spin chain with the asymmetry parameter $\Delta=-1/2$: evaluation of the simplest correlators}, Theor.\ Math.\ Phys.\ 129 (2001),  1596-1608.
\bibitem[Su1]{su1}
B.\ I.\ Suleimanov, 
\emph{The Hamilton property of Painlev\'e equations and the method of isomonodromic deformations},
Differ.\ Equ.\ 30 (1994), 726--732. 
\bibitem[Su2]{su2}
 B.\ I.\ Suleimanov, 
\emph{``Quantum" linearization of Painlev\'e equations as a component of their $L,\,A$ pairs}, Ufa Math.\ J.\ 4 (2012),  127--136.
\bibitem[TV]{tv} A.\ Treibich and J.-L.\ Verdier, 
\emph{Rev\^etements tangentiels et sommes de $4$ nombres triangulaires},
C.\ R.\ Acad.\ Sci.\ Paris S\'er.\ I Math.\ 311 (1990), 51--54. 
\bibitem[V]{v} A.\ P.\ Veselov, \emph{On Darboux--Treibich--Verdier potentials}, Lett.\ Math.\ Phys.\ 96, 209--216 (2011).
\bibitem[WW]{ww} E.\ T.\ Whittaker and G.\ N.\ Watson, A Course of Modern Analysis, 4th ed., Cambridge University Press, 1927.
\bibitem[ZZ]{zz}
A.\ Zabrodin and A.\ Zotov,
\emph{
Quantum Painlev\'e--Calogero correspondence for Painlev\'e VI},
J.\ Math.\ Phys.\ 53 (2012),  073508.
\bibitem[ZJ]{zj} P.\ Zinn-Justin, 
\emph{Sum rule for the eight-vertex model on its combinatorial line}, in Symmetries, Integrable Systems and Representations,  Springer, 2013, 599--637.
\vskip 3mm
\end{thebibliography}
 \end{document}